\documentclass[11pt,a4paper]{article}

\usepackage{amsmath}
\usepackage[T1]{fontenc}
\usepackage[utf8]{inputenc}
\usepackage{amsfonts}
\usepackage{amsthm}
\usepackage{amstext}
\usepackage{amssymb}
\usepackage{color}
\usepackage{bbold}
\usepackage{hyperref}
\usepackage{url}

\newcommand{\cc}{\mathbb{C}}
\newcommand{\rr}{\mathbb{R}}
\newcommand{\qq}{\mathbb{Q}}

\newcommand{\nn}{\mathbb{N}}
\newcommand{\kk}{\mathbb{K}}

\newcommand{\diag}{\operatorname{diag}}

\newtheorem{question}{Question}

\newtheorem{theorem}{Theorem}
\newtheorem{example}[theorem]{Example}

\newtheorem{lemma}[theorem]{Lemma}
\newtheorem*{lemma*}{Lemma}
\newtheorem{proposition}[theorem]{Proposition}
\newtheorem{corollary}[theorem]{Corollary}
\newtheorem{remark}[theorem]{Remark}

\newtheorem*{open*}{Open~question}
\newtheorem{definition}[theorem]{Definition}

\newcommand{\pnum}{p}
\newcommand{\pdenom}{q}
\newcommand{\hypar}{\mathcal{H}}
\newcommand{\NaN}{\mathrm{NaN}}
\newcommand{\pconf}{\mathcal{P}}

\newcommand{\hset}{\Lambda}

\begin{document}

\title{Derandomization and Absolute Reconstruction\\
  for Sums of Powers of Linear Forms}

  \author{Pascal Koiran\thanks{Univ Lyon, EnsL, UCBL, CNRS,  LIP, F-69342, LYON Cedex 07, France. { Email: pascal.koiran@ens-lyon.fr.}} \and Mateusz Skomra\thanks{{ LAAS-CNRS, Universit{\'e} de Toulouse, CNRS, Toulouse, France. Email: mateusz.skomra@laas.fr.}}}

  \maketitle

\begin{abstract}
  { We study the decomposition of multivariate polynomials as sums of powers
  of linear forms.}
   {  As one of our main results,  we give} a randomized algorithm for the following problem: given a homogeneous polynomial $f(x_1,\ldots,x_n)$ of degree 3, decide whether it can be written as a sum of cubes of linearly independent linear forms with complex coefficients.
  Compared to previous algorithms for the same problem,
  the two main novel features of this algorithm are:
  \begin{itemize}
\item[(i)] It is an algebraic algorithm, i.e., it performs only arithmetic operations and equality tests on the coefficients of the input polynomial $f$. In particular, it does not make any appeal to polynomial factorization.

\item[(ii)] For $f \in \qq[x_1,\ldots,x_n]$, the algorithm runs in polynomial time when implemented in the bit model of computation.
  \end{itemize}
  The algorithm relies on methods from linear and multilinear algebra (symmetric
  tensor decomposition by simultaneous diagonalization).
  We also give a version of our  algorithm for decomposition over the field of real numbers. In this case, the  algorithm performs arithmetic operations and
  comparisons on the input coefficients.
  
 Finally we give several related derandomization results on black box polynomial identity testing,
 the minimization of the number of variables in a polynomial,
 the computation of Lie algebras
  and factorization into products of linear forms.
  \end{abstract}

\newpage

\section{Introduction}

Let $f(x_1,\ldots,x_n)$ be a homogeneous polynomial of degree $d$, also called a {\em degree $d$ form.}  In this paper we study decompositions of the type:
\begin{equation} \label{eq:linind}
  f(x_1,\ldots,x_n)=\sum_{i=1}^r l_i(x_1,\ldots,x_n)^d
\end{equation}
where the $l_i$ are linear forms. Such a decomposition is sometimes called
a Waring decomposition, or a symmetric tensor decomposition.
We focus on the case where the linear forms $l_i$ are linearly independent.
This implies that the number $r$ of terms in the decomposition
is at most $n$. When $r=n$ we have
$f(x)=P_d(Ax)$ where $A$ is an invertible matrix of size $n$ and
\begin{equation} \label{eq:sumofpowers}
  P_d(x_1,\ldots,x_n) = x_1^d + x_2^d + \cdots +x_n^d
\end{equation}
is the ``sum of $d$-th powers'' polynomial.
If $f$ can be written in this way, we say that $f$ is equivalent to a
sum of $n$ $d$-th powers. More generally, two polynomials $f,g$ in $n$ variables
are said to be equivalent if they can be obtained from each other by
an invertible change of variables, i.e., if $f(x)=g(Ax)$ where $A$ is an
invertible matrix of size $n$. 
As pointed out in~\cite{kayal11}, the case $d=3$ (equivalence to a sum of $n$ cubes)
can be tackled with the decomposition algorithm for cubic forms in Saxena's thesis~\cite{Saxenathesis}.
Equivalence to $P_d$ for arbitrary $d$ was studied by
Kayal~\cite{kayal11}. This paper also begins a study of
equivalence to other specific polynomials such as the elementary symmetric polynomials;
this study is continued in~\cite{garg19,grochow12,grochow19,kayal2012affine,kayal18},
in particular for the permanent and determinant polynomials.
The contributions of the present paper are twofold:
\begin{itemize}
\item[(i)] We give efficient tests for equivalence to a sum of $n$ cubes over the fields of real and complex numbers. In particular, for an input polynomial
  with rational coefficients we give the first polynomial time algorithms
  in the standard Turing machine model of computation.
  As explained below in Section~\ref{subsec:equiv}, this is not in contradiction
  with the polynomial time bounds from~\cite{kayal11,Saxenathesis}
  because we do not address exactly the same problem or work in
  the same computation model as these two papers.

  More generally, we can test efficiently
  whether the input~$f$ can be written
  as in~(\ref{eq:linind}) as a sum of cubes of linearly independent linear forms.
  This follows easily from our equivalence tests
  and the algorithms from~\cite{kayal11,Saxenathesis} for the minimization
  of the number of variables in a polynomial.

\item[(ii)] Our first equivalence algorithm for the fields of real and complex
  numbers is randomized.
  We derandomize this algorithm in Section~\ref{sec:deter}, and we continue
  with 
  {several} related derandomization results on black box polynomial identity
  testing, 
  the minimization of the number of variables in a polynomial, the computation of Lie algebras
  and factorization into products of linear forms.
\end{itemize}

\subsection{Equivalence to a sum of cubes} \label{subsec:equiv}

Our algorithm for equivalence to a sum of $n$ cubes over $\cc$ is
{\em algebraic} in the sense that the input polynomial may have arbitrary
complex coefficients and we manipulate them using only arithmetic operations
and equality tests.
Over the field of real numbers we also allow inequality
tests. We therefore work in the ``real RAM'' model; an appropriate formalization
is provided by the { Blum-Shub-Smale} model of computation~\cite{BCSS,BSS89}.
We can provide algebraic algorithms only because we are considering a
{\em decision problem}: it is easy to see that if the input $f(x_1,\ldots,x_n)$
is equivalent to a sum of $n$ cubes, the coefficients of
the linear forms $l_i$ in the corresponding decomposition need
not be computable from those of $f$ by arithmetic operations
(see Examples~\ref{ex:noreal} and~\ref{ex:norat} at the beginning of Section~\ref{sec:equiv}).

Polynomial factorization is an important subroutine
in many if not most reconstruction algorithms for arithmetic circuits,
see e.g.~\cite{GKP17ISSAC,GKP18,karnin09,kayal11,kayal18,kayal19,shpilka09}.
It may even seem unavoidable for some problems: reconstruction of $\Pi \Sigma$
circuits is nothing but the problem of factorization into products of linear
forms, and reconstruction of $\Pi \Sigma \Pi$ circuits is factorization
into products of sparse polynomials.
Useful as it is, polynomial factorization is clearly not feasible
with arithmetic operations only, even for polynomials of degree 2.
We therefore depart from the aforementioned algorithms by avoiding
all use of such a subroutine.

For an input polynomial $f$ with rational coefficients, our  algebraic algorithms run
in polynomial time in the standard bit model of computation, i.e.,
they are ``strongly polynomial'' algorithms
(this is not automatic due to the issue
of coefficient growth during the computation).
We emphasize that even for an input $f \in \qq[x_1,\ldots,x_n]$ we are still considering
the problem of equivalence to a sum of $n$ cubes over the real or complex
numbers. Consider by contrast Kayal's equivalence algorithm~\cite{kayal11},
which appeals to a polynomial factorization subroutine. We can
choose to factor polynomials over, say, the field of rational numbers.
We can then run Kayal's algorithm without any difficulty
on a probabilistic polynomial time Turing machine, but the algorithm
will then reject the polynomial of Example~\ref{ex:norat} whereas our
algorithm will accept it.\footnote{Alternatively, one can run Kayal's algorithm
  in a computation model over $\cc$ where, in addition to arithmetic operations over complex numbers, root finding of univariate polynomials is considered
  an atomic operation (as suggested in a footnote of~\cite{garg19}).
  The algorithm would then give the same answer as our algorithm, but
  it would 
  not operate anymore  within the Turing machine model
  (or  within the BSS model).}
At first sight this difficulty seems to have a relatively simple solution: for
an input with rational coefficients, instead of factoring polynomials
in $\qq[X]$ we will factor in a field extension of $\qq$ containing the coefficients of the linear forms $l_i$ (for instance
in $\qq[\sqrt{2}]$ for Example~\ref{ex:norat}).
It is unfortunately not clear that this approach yields a polynomial time
algorithm because
it might lead to computations in a field extension of exponential degree.
We explain this point in more detail in Section~\ref{subsec:kayal}.
For the same reason (reliance on a polynomial factorization subroutine)
similar issues arise in the analysis of Saxena's 
decomposition  algorithm. A complete analysis of these two algorithms  for equivalence to a sum of powers over $\cc$ in the Turing machine model would entail good control of coefficient growth
and good bounds on  the degrees of the field extensions involved. This has not been done yet to the best of our knowledge.

\subsection{Derandomization}

We give a deterministic black box identity testing algorithm for polynomials
which can be represented as in~(\ref{eq:linind})
as a sum of powers  of linearly independent linear forms.
As we will see in Section \ref{sec:fewer},
the problem is really to decide whether the (unknown) number  of terms $r$
in the decomposition is equal to 0. 
Indeed, for $r \geq 1$ such a polynomial can never be identically zero (and the PIT problem for this family of polynomials 
can therefore be solved by a trivial algorithm in the white box model).
In contrast to our equivalence algorithms, this black box PIT
applies to homogeneous polynomials of arbitrary degree.

There is already a significant amount of work on identity testing for sums of powers of linear forms. In particular, Saxena \cite{Saxena08} gave a polynomial time algorithm in the white box model (where we have access to an arithmetic circuit computing the input polynomial). Subsequently, several algorithms were given for the black box model 
{  \cite{agrawal13,forbes13a,forbes13b,forbes14} but they do not run in polynomial time.
The current state of the art is in~\cite{forbes14}, with a black box algorithm running in time $s^{O(\log \log s)}$.
}
We obtain here a polynomial running time under the assumption that the $l_i$ are linearly independent.
Without this assumption,  designing a black box PIT algorithm running in polynomial time remains to the best of our knowledge an open problem.

In Section~\ref{sec:essential} we build on our black box PIT to derandomize
Kayal's algorithm for the minimization of the number of variables
in a polynomial~\cite{kayal11}. Like our black box PIT, this result applies
to polynomials that can be written as sums of powers of linearly independent
linear forms.
For such a polynomial, the minimal (or "essential") number of variables is just the number $r$ of terms in the corresponding decomposition~(\ref{eq:linind}).
{ We continue with the computation of Lie algebras of products of linear forms. Finally, our deterministic algorithm for this task is applied to
  the derandomization of a factorization algorithm from~\cite{koiran2018orbits} and of Kayal's algorithm for equivalence to~$P_d$~\cite{kayal11}.}

\subsection{Our approach}

We obtain our equivalence algorithms by viewing the coefficients of the input polynomial $f(x_1,\ldots,x_n)$  as the coefficients of a symmetric tensor 
$T$ of size~$n$ and order 3 (since $f$ is of degree 3).  
 Equivalence to a sum of $n$ cubes then amounts to a kind of diagonalizability property of $T$.
This approach is explained in detail in Section \ref{sec:equiv}. It can be viewed as a continuation of previous work on orthogonal tensor decomposition~\cite{koiran2019ortho}
(the present paper is more algorithmic, is not limited to orthogonal decompositions and can be read independently from~\cite{koiran2019ortho}).

We work on a tensor of size $n$  by cutting it into $n$ "slices"; each slice is a symmetric matrix of size $n$. We therefore rely on methods from linear algebra.
This explains the presence of a section of preliminaries on simultaneous reduction by congruence (which is then 
applied to the slices of $T$).
Despite these rather long preliminaries, the resulting randomized algorithm is remarkably simple: it is described in just 3 lines at the beginning of Section~\ref{sec:random}.

Our deterministic algorithms also rely on important insights from Kayal's paper~\cite{kayal11}. In particular we rely on the factorization
properties of the Hessian determinant of the input $f$, which we manage to use without appealing explicitly to a factorization subroutine
(as explained in Section \ref{subsec:equiv}, this is ruled out in our approach). Our deterministic algorithm for the minimization of the number 
of variables is directly inspired by the randomized algorithm for this problem in the same paper.

{   \subsection{Comparison with previous tensor decomposition algorithms}

There is a vast literature on tensor decomposition algorithms, most of them numerical (see \cite{beltran19} for a recent paper showing that many of these algorithms are numerically unstable).
From this literature,  two papers by De Lathauwer et al. \cite{lathauwer06,lathauwer04} are closely related to the present work since they already recognized the importance of simultaneous diagonalization by congruence for tensor decomposition. 
{ Jennrich's algorithm is also closely related for the same reason. One can find a presentation of this algorithm in the recent book by Moitra~\cite{moitra18}  but it goes back much further to Harschman~\cite{harshman70}, where it is presented as a uniqueness result}.
There are nevertheless important differences between the settings of  \cite{lathauwer06,lathauwer04,moitra18} and of the present paper. In particular, these three works 
do not phrase tensor decomposition 
as a decision problem but as an optimization problem which is solved by numerical means (one suggestion from \cite{lathauwer06,lathauwer04}  is to perform the simultaneous diagonalization with the extended QZ iteration from \cite{vanderveen96}; { Jennrich's algorithm as presented in~\cite{moitra18} relies on pseudoinverse computations and eigendecompositions). 
All these numerical algorithms attempt to produce a decomposition of a tensor $\tilde{T}$ which is close to the input tensor $T$.
If one tries to adapt them to the setting of the present paper there is a fundamental difficulty: given $\tilde{T}$ and its
decomposition, it is not clear how we can decide whether or not $T$ admits an exact decomposition.
This is the main reason why we need to design a new algorithm. As an alternative, one could attempt to run Jennrich's algorithm in symbolic mode. In particular, eigenvalues and the components of eigenvectors would be represented symbolically
as elements
of a field extension. This leads exactly to the same difficulty as with Kayal's algorithm: as explained in Section~\ref{subsec:equiv} and in more detail in Section~\ref{sec:equiv}, the resulting algorithm might not run in polynomial
time because it might lead to computations in field extensions of exponential degree.}}
Note finally that  \cite{lathauwer06,lathauwer04,moitra18} deal with decompositions of ordinary rather than symmetric tensors. 
Algorithms for symmetric tensor decomposition can be found in the algebraic literature, see e.g. \cite{brachat10,bernardi11}.
Like \cite{lathauwer06,lathauwer04}, these two papers do not provide any complexity analysis for their algorithms.

\subsection{Future work}

In the current literature there is apparently no polynomial time algorithm (deterministic or randomized)  in the Turing machine model for the following problem: given a homogeneous polynomial $f(x_1,\ldots,x_n)$ of degree 
$d \geq 4$ with rational coefficients, decide whether it is equivalent to $x_1^d+\cdots+x_n^d$ 
over $\cc$. 
{  It will be shown in a forthcoming paper that this can be done by extending the tensor-based approach of
the present paper to higher degree. This results in a black box algorithm with running time polynomial in $n$ and $d$~\cite{KS21}}.
Alternatively, one could try to modify Kayal's equivalence algorithm~\cite{kayal11} or provide a better analysis of the existing algorithm.
As we have argued in Section~\ref{subsec:equiv} this has not been done at present even for degree~3.
One could also try an approach based on Harrison's work~\cite{harrison75} like in Saxena's thesis~\cite{Saxenathesis}.

More generally, we suggest to pay more attention 
to {\em absolute circuit reconstruction}, i.e., arithmetic circuit reconstruction over $\cc$.\footnote{The name is borrowed from absolute factorization, a well studied problem in computer algebra (see e.g. \cite{cheze05,ChezeLecerf07,gao03,shaker09}).}
Circuit reconstruction over $\qq$ or over finite fields has a number-theoretic flavour, whereas circuit reconstruction over $\rr$ or $\cc$ is of a more geometric nature.

One goal could be to obtain algebraic decision algorithms; as we have explained, this requires the removal of all polynomial factorization subroutines.
In principle, this is always possible since the set of polynomials computable by arithmetic circuits of a given shape and and size is definable by polynomial (in)equalities, i.e., it is a constructible set (over $\cc$) or a semi-algebraic set (over $\rr$).\footnote{This argument does not provide 
by itself an efficient decision algorithm.}
Another goal would be to obtain good complexity bounds for the Turing machine model when they are not available in the existing literature.

{  \subsection{Organization of the paper}

Section~\ref{sec:prelim} is devoted to  preliminaries on simultaneous diagonalization, by similarity transformations and especially by congruence.
In Section~\ref{sec:equiv} we  review Kayal's equivalence algorithm and explain why it does not yield a polynomial time bound in the bit model of computation. We also present our tensor-based approach to the equivalence problem.
In Section~\ref{sec:random} we give a polynomial time randomized algorithm for equivalence to a sum of cubes based on 
this approach. We derandomize this algorithm in Section~\ref{sec:deter}. 
In Section~\ref{sec:PIT} we give a PIT algorithm for polynomials that can be written as sums of 
$d$-th powers of linearly independent linear forms. As mentioned before, our algorithm runs in 
polynomial time in the black box model.
We give a randomized algorithm to decide whether a polynomial can be expressed under that form in Section~\ref{sec:essential} (Proposition~\ref{prop:fewdecomp}).
The remainder of Section~\ref{sec:essential} is devoted to the derandomization of  several algorithms from~\cite{kayal11,kayal2012affine} related to sums of powers of linear forms.
 We begin in Section~\ref{sec:depend}
  with the computation of linear dependencies between polynomials.
  Then we give applications to the minimization of the number of variables in sums of powers of linear forms
  (in Section~\ref{sec:min}), and to the computation of Lie algebras of products of linear forms (in Section~\ref{sec:Lie}).
  This leads to the derandomization of a factorization algorithm from~\cite{koiran2018orbits} and of the equivalence
  algorithm by Kayal~\cite{kayal11} described in Section~\ref{subsec:kayal}.
  The resulting equivalence algorithm runs in polynomial time for every fixed value of the  degree $d$.
  For the computation of the Lie algebras of products of linear forms, we give as an intermediate result (see Proposition~\ref{hset} and Remark~\ref{rem:hset}) an identity testing algorithm for a certain class of rational functions (rather than polynomials as is done usually). In the commutative setting, this is to the best of our knowledge the first result of this type.
  }

\section{Preliminaries} \label{sec:prelim}

This section is devoted to preliminaries from linear algebra, and more specifically to simultaneous diagonalization by congruence. We begin with complex symmetric 
matrices in Section \ref{sec:cmat} and consider real symmetric matrices in Section \ref{sec:realmat}. Section \ref{sec:refine} is devoted to some refinements that are not strictly necessary for our main algorithms (they lead to an interesting connection with semidefinite programming, though; see Theorem \ref{th:pd} and the remarks following it).
Upon first reading, if one wishes to understand only our results for the field of complex numbers it will therefore be sufficient to read Section \ref{sec:cmat} only (or even just the statement of Theorem \ref{th:simdiag}; the corresponding result for the field of real numbers is Theorem \ref{th:realsimdiag}).

A proof of the following lemma for the case of two matrices can be found in~\cite{koiran2018orbits}. As shown in~\cite{koiran2019ortho}, the general case then follows easily.
Note that this lemma is about the "usual" notion of diagonalization (by similarity) rather than by congruence (where one attempts to diagonalize a matrix  $A$ by a transformation
of the form $A \mapsto R^TAR$).
That is, we say that an invertible matrix $T$ diagonalizes $A$ if $T^{-1}AT$
is diagonal.
\begin{lemma} \label{lincomb}
  Let $A_1,\ldots,A_k \in M_n(\kk)$ be a tuple of simultaneously diagonalizable
  matrices with entries in a field $\kk$, and let $S \subseteq \kk$ be a finite set of size $|S| > n(n-1)/2$.
  Then there exist $\alpha_2,\ldots,\alpha_k$ in $S$ such that any transition
  matrix which diagonalizes $A_1+\alpha_2 A_2 + \ldots + \alpha_k A_k$
  must also diagonalize all of the matrices $A_1,\ldots, A_k$.
\end{lemma}
See Proposition \ref{prop:center} in Section \ref{sec:realmat} for an improvement of this lemma.
{ 
In the next lemma we give an explicit description of a set of  suitable choices for the tuple $(\alpha_2,\ldots,\alpha_k)$ 
in Lemma~\ref{lincomb}. This description will be needed for a derandomization result in Section~\ref{sec:Lie}.
\begin{lemma} \label{lincomb2}
  Let $A_1,\ldots,A_k \in M_n(\kk)$ be a tuple of simultaneously diagonalizable
  matrices with entries in a field $\kk$. There exists a set of at most $n(n-1)/2$ hyperplanes of $\kk^{n-1}$ such that the following properties hold for any
  point $(\alpha_2,\ldots,\alpha_{n-1}) \in \kk^{n-1}$ which does not belong to the union of the hyperplanes: 
  \begin{itemize}
\item[(i)] Any eigenvector of 
  $A_1+\alpha_2A_2+\cdots+\alpha_k A_k$ must also be an eigenvector of the $k$ matrices $A_1,\ldots,A_k$.
  \item[(ii)]  Any transition
  matrix which diagonalizes $A_1+\alpha_2 A_2 + \ldots + \alpha_k A_k$
  must also diagonalize all of the matrices $A_1,\ldots, A_k$.
  \end{itemize}
\end{lemma}
\begin{proof}
Since  $A_1,\ldots,A_k$ are simultaneously diagonalizable, in order to establish (i) we may as well work in a basis where these matrices become
diagonal. Let us therefore assume without loss of generality that $A_i=\diag(\lambda_{1i},\ldots,\lambda_{ni})$.
For any $\alpha_2,\ldots,\alpha_{n-1}$, the matrix $A_1+\alpha_2A_2+\cdots+\alpha_k A_k$ is diagonal and its 
$i$-th diagonal entry is $\lambda_{i1}+\alpha_2 \lambda_{i2}+\cdots+\alpha_k \lambda_{ik}$.
It therefore suffices to avoid the (proper) hyperplanes of $ \kk^{n-1}$ of the form:
$$\lambda_{i1}+\alpha_2 \lambda_{i2}+\cdots+\alpha_k \lambda_{ik} = \lambda_{j1}+\alpha_2 \lambda_{j2}+\cdots+\alpha_k \lambda_{jk}$$
where $(i,j) \in [n]^2$ ranges over all the pairs such that 
$(\lambda_{i1},\ldots,\lambda_{ik}) \neq (\lambda_{j1},\ldots,\lambda_{jk})$.
Indeed, there are at most $n(n-1)/2$ hyperplanes of this form, and outside of their union there is no ``collision of eigenvalues.'' More precisely, the eigenspace of $A_1+\alpha_2A_2+\cdots+\alpha_k A_k$ associated to the eigenvalue 
$\lambda_{i1}+\alpha_2 \lambda_{i2}+\cdots+\alpha_k \lambda_{ik}$ is equal to $\bigcap_{j=1}^k E_{ij}$, 
where $E_{ij}$ denotes the eigenspace of $A_j$ associated to $\lambda_{ij}$. 
In particular, any eigenvector of  $A_1+\alpha_2A_2+\cdots+\alpha_k A_k$ is also an eigenvector of  $A_1,\ldots, A_k$.

Finally, we show that any point $(\alpha_2,\ldots,\alpha_{n-1})$ which satisfies (i) also satisfies (ii). For any invertible matrix $T$, $T^{-1}(A_1+\alpha_2A_2+\cdots+\alpha_k A_k)T$ is diagonal iff all the column vectors of $T$ are eigenvectors 
of $A_1+\alpha_2A_2+\cdots+\alpha_k A_k$. By (i) this implies that each column vector of $T$ is also an eigenvector
of $A_1,\ldots,A_k$. As a result, the $k$ matrices $T^{-1}A_i T$ are all diagonal.
\end{proof}
We note that Lemma~\ref{lincomb} directly follows from Lemma~\ref{lincomb2}.(ii) via e.g. the Schwartz-Zippel lemma.
}

\subsection{Simultaneous diagonalization by congruence} \label{sec:cmat}

The following result is from Horn and Johnson \cite{horn13}. The first part is just the statement of Theorem 4.5.17(b), and the additional properties in (ii) are established 
in the proof of that theorem (see  \cite{horn13} for details).
\begin{theorem} \label{4.5.17}
Let $A,B \in M_n(\cc)$ be two complex symmetric matrices of size~$n$ with $A$ nonsingular, and let $C=A^{-1}B$.
\begin{itemize}
\item[(i)] {  $C$ is diagonalizable if and only if there are complex diagonal matrices $D$ and $\Delta$ and a nonsingular $R \in M_n(\cc)$  such that $A=R D R^T$ and $B=R \Delta R^T$.}
\item[(ii)] Moreover, if $C= S \Lambda S^{-1}$ where $S$ is nonsingular and $\Lambda$ diagonal then the matrix $R$ in (i) can be taken of the form $R=S^{-T}V^{T}$ where $V$  is unitary and commutes with $\Lambda$.
\end{itemize}
\end{theorem}
The next result generalizes Theorem \ref{4.5.17} and provides a solution to the second part of Problem~4.5.P4 in~\cite{horn13}. 
\begin{theorem}[simultaneous diagonalization by congruence] \label{th:simdiag}
Let $A_1,\ldots,A_k$ be complex symmetric matrices of size $n$ and assume that $A_1$ is nonsingular. 
{  The $k-1$ matrices $A_1^{-1}A_i$ ($i=2,\ldots,k$) form a commuting family of diagonalizable matrices  if and only if there are   diagonal matrices $\Lambda_i$ and a  nonsingular matrix $R \in M_n(\cc)$ such that 
$A_i = R \Lambda_i R^T$ for all $i=1,\ldots,k$.}
\end{theorem}
\begin{proof}
Suppose that $A_i = R \Lambda_i R^T$ where the $\Lambda_i$ are diagonal and $R$ nonsingular. 
Then the matrices $A_1^{-1}A_i = R^{-T}\Lambda_1^{-1} \Lambda_i R^T$ indeed form a commuting family of diagonalizable matrices.
For the converse, assume that the matrices $A_1^{-1}A_i$  form such a family. 
Then these matrices are simultaneously diagonalizable, and by Lemma~\ref{lincomb} there is a tuple $(\alpha_3,\ldots,\alpha_k)$ such that any transition matrix that diagonalizes the matrix
$$C=A_1^{-1}A_2+ \alpha_3 A_1^{-1}A_3+ \ldots+\alpha_k A_1^{-1}A_k$$ diagonalizes all of the $k-1$ matrices $C_i=A_1^{-1}A_i$.
Now we apply Theorem~\ref{4.5.17} to $A=A_1$ and $B=A_2 + \alpha_3 A_3+\ldots+\alpha_k A_k$.
Write $C=S \Lambda S^{-1}$ where $S$ is nonsingular and $\Lambda$ diagonal. By part (ii) of Theorem \ref{4.5.17} we can write  $A_1=R D R^T$ and $B=R \Delta R^T$ where $R$ is nonsingular, $D$ and $\Delta$ are diagonal,  $R=S^{-T}V^{T}$, $V$ commutes with $\Lambda$ and is unitary.
By choice of the tuple~$\alpha$, we can write $C_i=S \Lambda_i S^{-1}$ where $ \Lambda_i$ is diagonal. We will show that { $A_i = R D \Lambda_i R^T$} for $i \geq 2$, thereby completing the proof of the theorem.
First, we note that $V$ commutes with the $ \Lambda_i$. Indeed, $V$ and $\Lambda$ are simultaneously diagonalizable since they commute and are diagonalizable ($\Lambda$  is diagonal and $V$ unitary).
But any transition matrix which diagonalizes simultaneously $V$ and $\Lambda$ will diagonalize simultaneously $V$ and the $ \Lambda_i$ (this follows from the choice of $\alpha$ and the relations  $C_i=S \Lambda_i S^{-1}$, $C=S \Lambda S^{-1}$).
These matrices must therefore commute.
We can now complete the proof:  for $i \geq 2$ we have 
 $$A_i=A_1C_i = (S^{-T}V^{T} D VS^{-1})( S\Lambda_i S^{-1}) = S^{-T}V^{T} D V\Lambda_i S^{-1}.$$ 
Since $V$ commutes with $\Lambda_i$, 
{ $A_i=S^{-T}V^{T} D \Lambda_i V S^{-1} = R D\Lambda_i R^T$} as announced.
\end{proof}

\subsection{A refinement of Theorem \ref{th:simdiag}} \label{sec:refine}

In this section we give a more "invariant" formulation of Theorem \ref{th:simdiag}. Note indeed that this theorem assigns a special role to $A_1$. 
In Theorem \ref{th:dissociate} we give a formulation that depends only {  on}  the space spanned by the $A_i$ and not on the choice of a specific spanning family 
$A_1,\ldots,A_k$. Some of the results in this section apply to $\kk=\rr$ as well as $\kk=\cc$.

The role of $A_1$ in Theorem \ref{th:simdiag} could of course be played by any other invertible matrix in the tuple. As it turns out, for our $k-1$ matrices the commutation property alone is also independent
of the choice of the invertible matrix in the tuple. 
More precisely, we have:
\begin{proposition} \label{prop:exchange}
Let $A_1,\ldots,A_k \in M_n(\kk)$; assume that $A_1$ and $A_k$ are nonsingular. 
The $k-1$ matrices  $A_k^{-1}A_i$ ($i=1,\ldots,k-1$) commute if and only if the same is true of the  $k-1$ matrices $A_1^{-1}A_i$ ($i=2,\ldots,k$).
\end{proposition}
The proof will use the following simple fact.
\begin{lemma} \label{lem:invcom}
If $A$ and $B$ commute and $A$ is invertible, then $A^{-1}$ and $B$ commute as well.
\end{lemma}

\begin{proof}[Proof of Proposition \ref{prop:exchange}]
Suppose that the $A_k^{-1}A_i$ commute. We can write:
$$(A_1^{-1}A_i) (A_1^{-1}A_j) =  (A_k^{-1}A_1)^{-1}  (A_k^{-1}A_i) (A_k^{-1}A_1)^{-1} (A_k^{-1}A_j).$$ 
It follows from our hypothesis and from Lemma \ref{lem:invcom}  that the four factors on the right hand side commute.
We can therefore rewrite this equation as:
$$(A_1^{-1}A_i) (A_1^{-1}A_j) =  (A_k^{-1}A_1)^{-1}  (A_k^{-1}A_j) (A_k^{-1}A_1)^{-1} (A_k^{-1}A_i),$$ 
and now the right hand side is equal to $(A_1^{-1}A_j) (A_1^{-1}A_i)$.
\end{proof}
Let $\cal V$ be the space of  matrices spanned by $A_1,\ldots,A_k$. The matrices $A_1^{-1}A_i$ commute if and only if $A_1^{-1} {\cal V}$ is a commuting subspace of $M_n(\kk)$. With Proposition \ref{prop:exchange} in hand, we can characterize this property in a way that is completely independent of the choice of a spanning family $A_1,\ldots,A_k$ for $\cal V$.

\begin{theorem} \label{th:invariant}
Let $\cal V$ be a  nonsingular  subspace of matrices of $M_n(\kk)$ (i.e., $\cal V$ does not contain singular matrices only).
The two following properties are equivalent:
\begin{itemize}
\item[(i)] There exists a nonsingular matrix $A \in \cal V$ such that $A^{-1} {\cal V}$ is a commuting subspace.
\item[(ii)] For all nonsingular matrices $A \in \cal V$, $A^{-1} {\cal V}$ is a commuting subspace.
\end{itemize}
\end{theorem}
\begin{proof}
Since $\cal V$ is nonsingular, (ii) implies (i). For the converse, assume that $A_1^{-1} {\cal V}$ is a commuting subspace 
and that $A_1,\ldots,A_k$ is a spanning family of~$\cal V$. Let $A \in \cal V$ be a nonsingular matrix. 
We can add $A$ to our spanning family and apply Proposition \ref{prop:exchange} to $(A_1,\ldots,A_k,A)$, with $A$ playing the role of $A_k$ in that proposition.
\end{proof}
As a result, we can dissociate the commutativity test from the diagonalizability test in Theorem \ref{th:simdiag}:
\begin{theorem} \label{th:dissociate}
Let $A_1,\ldots,A_k$ be complex symmetric matrices of size $n$ and assume that the subspace $\cal V$ spanned by these matrices
is nonsingular.
{  The two following properties are equivalent:
\begin{itemize}
\item[(i)] There are   diagonal matrices $\Lambda_i$ and a  nonsingular matrix $R \in M_n(\cc)$ such that 
$A_i = R \Lambda_i R^T$ for all $i=1,\ldots,k$. 
\item[(ii)] $\cal V$ satisfies the two equivalent properties of Theorem \ref{th:invariant}, and there exists an invertible  $B \in \cal V$ such that the matrices $B^{-1}A_i$ ($i=1,\ldots,k$) are all diagonalizable.
\end{itemize}
}
\end{theorem}
\begin{proof}
Let $A \in \cal V$ be nonsingular, and suppose that there   are   diagonal matrices $\Lambda_i$ and a  nonsingular matrix $R \in M_n(\cc)$ such that 
$A_i = R \Lambda_i R^T$ for all~$i$. Note that we have the same form for $A$, i.e., $A= R \Lambda R^T$ with $\Lambda$ diagonal. 
As a result, we may assume without loss of generality that $A$ is one of the matrices in the tuple $A_1,\ldots,A_k$ (we add it if necessary),
and we may even assume that $A=A_1$. We may then take $B=A$ by Theorem \ref{th:simdiag}.

Let us now prove the converse. We therefore assume that $\cal V$ satisfies the two properties of Theorem \ref{th:invariant}, 
and that $B \in \cal V$ is an invertible matrix such the matrices $B^{-1}A_i$ ($i=2,\ldots,k$) are all diagonalizable.
From property (ii) in Theorem \ref{th:invariant} it follows that the matrices $B^{-1}A_i$ commute. We conclude by applying Theorem \ref{th:simdiag}
to the tuple $(B,A_1,\ldots,A_k)$.
\end{proof}

\subsection{Real matrices} \label{sec:realmat}

Here we study the existence of decompositions similar to those of Theorem~\ref{4.5.17} and Theorem~\ref{th:simdiag} for real matrices. 
We begin with a real version of  Theorem~\ref{4.5.17}.
\begin{theorem} \label{real4.5.17}
Let $A,B \in M_n(\rr)$ be two real symmetric matrices of size~$n$ with $A$ nonsingular, and let $C=A^{-1}B$.
\begin{itemize}
\item[(i)] 
{  $C$ is diagonalizable and has real eigenvalues if and only if there are  real diagonal matrices $D$ and $\Delta$ and a nonsingular $R \in M_n(\rr)$ such that $A=R D R^T$ and $B=R \Delta R^T$.}
\item[(ii)] Moreover, if $C= S \Lambda S^{-1}$ where $S$ is a real nonsingular matrix and $\Lambda$ diagonal then the matrix $R$ in (i) can be taken of the form $R=S^{-T}V^{T}$ where $V$  is orthogonal and commutes with $\Lambda$.
\end{itemize}
\end{theorem}
\begin{proof}
If a decomposition of the pair $(A,B)$ as in (i) exists, it is clear that $C$ must be diagonalizable with real eigenvalues since $C=R^{-T}D^{-1}\Delta R^T$.
The converse and part (ii) can be obtained by a straightforward adaptation of the proof of Theorem~{4.5.17(b)} in \cite{horn13}.
\end{proof}

The next result generalizes Theorem \ref{real4.5.17} and provides a real version of Theorem \ref{th:simdiag}.
\begin{theorem}[simultaneous diagonalization by congruence] \label{th:realsimdiag}
Let $A_1,\ldots,A_k$ be real symmetric matrices of size $n$ and assume that $A_1$ is nonsingular. 
 {  The $k-1$ matrices $A_1^{-1}A_i$ ($i=2,\ldots,k$) form a commuting family of diagonalizable matrices with real eigenvalues
 if and only if there are   real diagonal matrices $\Lambda_i$ and a  nonsingular matrix $R \in M_n(\rr)$ such that 
$A_i = R \Lambda_i R^T$ for all $i=1,\ldots,k$.}
\end{theorem}
\begin{proof}
The proof of Theorem \ref{th:simdiag} applies almost verbatim: we just need to work everywhere with real matrices instead of complex matrices (and with real coefficients 
$\alpha_i$), and appeal to Theorem \ref{real4.5.17} instead of Theorem \ref{4.5.17}. 
There is just one point in the proof where a little care is needed.
Namely, in the proof of Theorem \ref{th:simdiag} we used the fact that Theorem \ref{4.5.17}.(ii) provides us with a unitary matrix $V$, and that unitary matrices are diagonalizable.
In the real case we get an orthogonal matrix instead  (as per Theorem \ref{real4.5.17}.(ii)), and orthogonal matrices are not necessarily diagonalizable over $\rr$.
Nevertheless, real orthogonal matrices are diagonalizable over $\cc$ since they are unitary. We can therefore conclude like in the proof of Theorem \ref{th:simdiag} 
that our real orthogonal matrix $V$ commutes with the $\Lambda_i$. The remainder of the proof is unchanged.
\end{proof}

In the proofs of Theorems \ref{th:simdiag} and \ref{th:realsimdiag} we have used the fact that $V$ is a unitary matrix. These arguments can be somewhat  simplified at the expense of proving the following improvement to Lemma \ref{lincomb}. First we recall that the centralizer of a matrix $M$, denoted $Z(M)$, is the subspace of matrices that commute with $M$.

\begin{proposition} \label{prop:center}
Let $A_1,\ldots,A_k$ and $\alpha_2,\ldots,\alpha_k$ be as in Lemma \ref{lincomb}; let $A=A_1+\alpha_2 A_2 + \ldots + \alpha_k A_k$.
Then $$Z(A)=\bigcap_{i=1}^k Z(A_i).$$
\end{proposition}
This result applies to real as well as complex matrices.
Before giving the proof, let us  explain how it can be used in Theorems \ref{th:simdiag} and \ref{th:realsimdiag}. Applying Proposition \ref{prop:center} to the tuple of simultaneously diagonalizable matrices $C_2,\ldots,C_k$, we see that  $$Z(C)=\bigcap_{i=2}^k Z(C_i).$$ Since  $C_i=S \Lambda_i S^{-1}$ and $C=S \Lambda S^{-1}$,
this implies $Z(\Lambda)=\bigcap_{i=2}^k Z(\Lambda_i).$
Since $V \in Z(\Lambda)$, we conclude that $V$ commutes with the $\Lambda_i$. Therefore, we have established this commutation property without using the fact that
$V$ can be taken unitary.

\begin{proof}[Proof of Proposition \ref{prop:center}]
The inclusion from right to left obviously holds for any choice of the~$\alpha_i$. For the converse, let $B \in Z(A)$ and assume as a first step that $B$
is diagonalizable. Since $A$ and $B$ commute and both matrices are diagonalizable, there exists a transition matrix $T$ such that $T^{-1}AT$ and $T^{-1}BT$ are diagonal. 
By choice of the $\alpha_i$, all of the matrices $T^{-1}A_iT$ are diagonal as well. We conclude that $B$ and $A_i$ commute since they are simultaneously diagonalizable.
To complete the proof, we just need to observe that diagonalizable matrices are dense in $Z(A)$. This follows from the fact that $A$ itself is diagonalizable
(observe indeed that the centralizer of a diagonal matrix takes a block-diagonal from, and diagonalizable matrices are dense in each block).
\end{proof}

Like in the complex case, we can dissociate the commutativity test in Theorem \ref{th:realsimdiag} from the diagonalizability test:
\begin{theorem} \label{th:realdissociate}
Let $A_1,\ldots,A_k$ be real symmetric matrices of size $n$ and assume that the subspace $\cal V$ spanned by these matrices
is nonsingular.
{  The two following properties are equivalent:
\begin{itemize}
\item[(i)] There are   diagonal matrices $\Lambda_i$ and a  nonsingular matrix $R \in M_n(\rr)$ such that
$A_i = R \Lambda_i R^T$ for all $i=1,\ldots,k$.
\item[(ii)]   $\cal V$ satisfies the two equivalent properties of Theorem \ref{th:invariant}, and there exists an invertible  $B \in \cal V$ such that the matrices $B^{-1}A_i$ ($i=2,\ldots,k$) are all diagonalizable with real eigenvalues.
\end{itemize}
}
\end{theorem}
The proof is identical to the proof of Theorem \ref{th:dissociate}, except that we appeal to Theorem \ref{th:realsimdiag} instead of Theorem \ref{th:simdiag}.
The criterion in Theorem \ref{th:realdissociate} takes a particularly simple form when $\cal V$ contains a positive definite matrix.
Before explaining this, we recall the following lemma.
\begin{lemma} \label{lem:pd}
Let $A$ and $B$ be two real symmetric matrices with $B$ positive definite. Then $B^{-1}A$ is diagonalizable with real eigenvalues.
\end{lemma}

\begin{proof}
Since $B$ is positive definite, we can write $B=H H^T$ where $H$ is a real invertible matrix. Hence $B^{-1}A = H^{-T}H^{-1} A = H^{-T} (H^{-1} A H^{-T}) H^T$.
Since $H^{-1} A H^{-T}$ is real symmetric it is diagonalizable with real eigenvalues. This is true of $B^{-1}A$ as well since the  two matrices are similar.
\end{proof}

As an immediate consequence of Lemma \ref{lem:pd} and Theorem \ref{th:realdissociate} 
we have:

\begin{theorem}[simultaneous diagonalization by congruence, positive definite case] \label{th:pd}
Let $A_1,\ldots,A_k$ be real symmetric matrices of size $n$ and assume that 
 the subspace $\cal V$ spanned by these matrices contains a positive definite matrix.
 The 3 following properties are equivalent:
 \begin{itemize}
 \item[(i)] There exists a nonsingular matrix $A \in \cal V$ such that $A^{-1} {\cal V}$ is a commuting subspace.
 \item[(ii)] For all nonsingular matrices $A \in \cal V$, $A_1^{-1} {\cal V}$ is a commuting subspace.
  \item[(iii)] There are   real diagonal matrices $\Lambda_i$ and a  nonsingular matrix $R \in M_n(\rr)$ such that
$A_i = R \Lambda_i R^T$ for all $i=1,\ldots,k$.
\end{itemize}
\end{theorem}
{A related characterization can be found in~\cite[Theorem~3.3]{jiang16}.}
As we will see at the end of  Section \ref{sec:equiv}, the significance of
Theorem~\ref{th:pd}
is that when a polynomial is equivalent to a sum of $n$ real cubes,
the corresponding 
$\cal V$ always contains a positive definite matrix.

\section{The equivalence problem} \label{sec:equiv}

In this section we review Kayal's equivalence algorithm
{  and present our  tensor-based approach.}
We first recall the following definition from the introduction.
\begin{definition} \label{def:equiv}
  A polynomial $f \in \kk[x_1,\ldots,x_n]$ is said to be equivalent to a sum of $n$ $d$-th powers 
  if $f(x)=P_d(Ax)$ where $P_d(x_1,\ldots,x_n)=x_1^d+\ldots+x_n^d$
and $A \in M_n(\kk)$ is a nonsingular matrix.
\end{definition}
As explained before, our equivalence algorithms
in Sections~\ref{sec:random} and~\ref{sec:deter} deal only with the case
$d=3$ (equivalence to a sum of cubes).
More generally, one could ask whether two forms of degree 3 given as input are equivalent (by an invertible change of variables as above).
This problem is known to be at least as hard as graph isomorphism~\cite{agrawal06,kayal11} and is "tensor isomorphism complete"~\cite{grochow19}.
By contrast, equivalence of quadratic forms is "easy": it is classically known from linear algebra that   two real quadratic forms are equivalent iff they have the same rank and signature (this is "Sylvester's law of inertia")  
and two quadratic forms over $\cc^n$ are equivalent iff they have the same rank.\footnote{This paper is focused on the fields of real and complex numbers but equivalence of quadratic forms has been an active topic of study for other fields as well,
especially from the point of view of number theory: see for instance \cite{omeara,scharlau}.}

Note that when $\kk=\rr$, Definition \ref{def:equiv} requires the changes of variables matrix $A$ as well as the input polynomial $f$ to be real.
It also makes sense to ask if an input $f \in \rr[x_1,\ldots,x_n]$ is equivalent to a sum of $n$ cubes as a complex polynomial.
The following example shows that these are two distinct notions of equivalence.
\begin{example} \label{ex:noreal}
  Consider the 
  real polynomial 
$$f(x_1,x_2)=(x_1+ix_2)^3+(x_1-ix_2)^3=2x_1^3-6x_1x_2^2.$$
This decomposition shows that as a complex polynomial, $f$ is equivalent to a sum of two cubes. Moreover, there is no decomposition as a sum of 2 cubes of {\em real} linear forms since the above decomposition is essentially the unique decomposition of $f$.  This follows from Corollary \ref{cor:unique} below:
in any other decomposition $f=l_1^3+l_2^3$, the linear forms $l_1$ and $l_2$ must be scalar multiples of $x_1+ix_2$ and  $x_1-ix_2$.\footnote{More precisely one must have $l_1=\alpha_1(x_1+ix_2)$, $l_2=\alpha_2(x_1-ix_2)$ (or vice versa) where $\alpha_1^3=\alpha_2^3=1$.}
Note that this is very different from the case of degree~2 forms: any real quadratic form in $n$ variables can be written as a linear combinations of $n$ real squares.
\end{example}
A similar example shows that for a polynomial $f(x_1,x_2)$
with rational coefficients, equivalence to a sum of two cubes over $\qq$ or
other $\rr$ are distinct notions:
\begin{example}  \label{ex:norat}
  Consider the rational  polynomial 
$$f(x_1,x_2)=(x_1+\sqrt{2}x_2)^3+(x_1-\sqrt{2}x_2)^3 = 2x_1^3+12x_1x_2^2.$$
This polynomial is equivalent to a sum of two cubes over $\rr$
but not over $\qq$.
\end{example}

\subsection{Review of Kayal's equivalence algorithm} \label{subsec:kayal}

Let $f \in \cc[x_1,\ldots,x_n]$ be a homogeneous polynomial
of degree $d \geq 3$. Recall that the Hessian matrix of $f$ is the symmetric
matrix of size $n$ with entries $\partial^2 f / \partial x_i \partial x_j$,
and that the Hessian determinant of $f$, denoted $H_f$, is the determinant
of this matrix.
Kayal's equivalence algorithm is based on the factorization properties
of the Hessian determinant. Note that $H_f=[d(d-1)]^n (x_1\cdots x_n)^{d-2}$
for $f=P_d$. It is shown in~\cite{kayal11} that we still have a factorization
as a product of linear forms after an invertible change of variables:
\begin{lemma} \label{lem:hessian}
 Suppose that $f(x_1,\ldots,x_n) = \sum_{i=1}^n a_il_i(x_1,\ldots,x_n)^d$
 where the $l_i$ are linear forms and $a_i \in \cc \setminus \{0\}$.
 The Hessian determinant of $f$ is of the form
 $$H_f(x_1\ldots,x_n)=c\prod_{i=1}^n l_i(x_1,\ldots,x_n)^{d-2}$$
 for some constant $c$.
 Moreover, $c \neq 0$ iff the $l_i$ are linearly independent.
\end{lemma}
As a consequence we have the uniqueness result of Corollary~\ref{cor:unique} below, which generalizes
Corollary~5.1 in~\cite{kayal11}. First, we need the following lemma:
\begin{lemma}  \label{lem:nonzero}
Suppose that $f$ can be written as in (\ref{eq:linind}) as a sum of powers of linearly independent linear forms. If $r \geq 1$ then $f$ is not identically zero.
\end{lemma}
\begin{proof}
  We already know this for $r=n$: Lemma \ref{lem:hessian} shows that $H_f$ is not identically 0. For a more direct proof of the result in this case,
  one can simply observe that $f$ is equivalent to $P_d$ but $P_d$
  is not equivalent to 0.
  
  The general case can be reduced to the case $r=n$
  by setting $n-r$ of the variables of $f$ to 0. 
In this way, we obtain a sum of powers of $r$ linear forms $l'_i$ in $r$ variables, Moreover, it is always possible to choose the variables of $f$ that are set to 0 so that the $l'_i$ remain linearly independent like the forms $l_i$ in (\ref{eq:linind}). Indeed, this follows from the fact that a $r \times n$ matrix
of rank $r$ must contain a $r \times r$ submatrix of rank $r$.
\end{proof}
\begin{corollary} \label{cor:unique}
  Suppose that $f(x_1,\ldots,x_n)=\sum_{i=1}^r l_i(x_1,\ldots,x_n)^d$
  where the $l_i$ are linearly independent linear forms.
  For any other decomposition
  $f(x_1,\ldots,x_n)=\sum_{i=1}^r \ell_i(x_1,\ldots,x_n)^d$
  the linear forms $\ell_i$ must satisfy $\ell_i = \omega_i l_{\pi(i)}$ where $\omega_{i}$ is a $d$-th root of unity and $\pi \in S_n$ a permutation.
\end{corollary}
\begin{proof} Consider first the case $r=n$. By Lemma~\ref{lem:hessian}
  and uniqueness of factorization we must have $\ell_i = c_i l_{\pi(i)}$
  for some constants $c_i$ and some permutation~$\pi$.
  Plugging this relation into the two decompositions of $f$ shows that:
  $$\sum_{i=1}^n l_i^d = \sum_{i=1}^n c_i^d l_{\pi(i)}^d.$$
  Moving all terms to the left-hand side we obtain:
  $$\sum_{i=1}^n (1-c_{\pi^{-1}(i)}^d)l_i^d = 0,$$
  and Lemma~\ref{lem:nonzero} then implies that $c_i^d = 1$ for all $i$.
  Assume now that $r <n$. Since the $l_i$ are linearly independent,
  we can extend this family into a family of $n$
  linearly independent linear forms $l_1,\ldots,l_n$.
    Our two decompositions of $f$ yield two decompositions for the polynomial
  $g=f+l_{r+1}^d+\cdots+l_n^d$, and we can apply the result
  for the case $r=n$ to $g$. Another way to reduce to this case would be to
  decrease $n$ by setting $n-r$ of the variables of $f$ to 0 as
  in the proof of Lemma~\ref{lem:nonzero}.
\end{proof}

Kayal's algorithm can be summarized by the 3 following steps. It takes as input a degree $d$ form $f \in K[x_1,\ldots,x_n]$  and determines whether $f$ is equivalent to $P_d$ over $\mathbb{K}$, where $K \subseteq \mathbb{K}$ are two subfields of $\cc$. If $f$  is equivalent to $P_d$ it determines linearly independent linear forms $\ell_i \in \mathbb{K}[x_1,\ldots,x_n]$ such that  $f(x_1,\ldots,x_n) = \sum_{i=1}^n \ell_i(x_1,\ldots,x_n)^d$.
This presentation generalizes slightly~\cite{kayal11}, which focuses on the case $K=\mathbb{K} = \cc$.
\begin{enumerate}
\item Check that the Hessian determinant $H_f$ is not identically 0 and can be factorized in $\mathbb{K}[x_1,\ldots,x_n]$ as $H_f(x_1\ldots,x_n)=c\prod_{i=1}^n l_i(x_1,\ldots,x_n)^{d-2}$ where the $l_i$ are linear forms and $c \in \mathbb{K}$.
If this is not possible, reject. 
\item Try to find constants $a_i \in \mathbb{K}$ such that  $$f(x_1,\ldots,x_n) = \sum_{i=1}^n a_il_i(x_1,\ldots,x_n)^d.$$ If this is not possible, reject. 

\item Check that all the $a_i$ have $d$-th roots in $\mathbb{K}$. If this is not the case, reject. Otherwise, declare that $f$ is equivalent to $P_d$ over $\mathbb{K}$
and output the linear forms $\ell_i=\alpha_i l_i$ where $\alpha_i^d=a_i$ and $\alpha_i \in \mathbb{K}$.
\end{enumerate}
The correctness of the algorithm follows from Lemma \ref{lem:hessian} and
Corollary~\ref{cor:unique}.
Note in particular that if the algorithm accepts, the forms $l_i$ must be linearly independent (or else $H_f$ would be identically 0 by Lemma~\ref{lem:hessian}, and the algorithm would have rejected at step 1); and the constants $a_i$ at
step 2 are unique if they exist.

For $d=3$, or more generally for small degree, the constants $a_i$ at step 2 can be found efficiently by dense linear algebra assuming an algebraic model of computation  (for the Turing machine model, see the comments below).
For large $d$ we can instead evaluate $f$ and the powers $l_i^d$ at random points (see Section \ref{sec:essential} on linear dependencies or \cite{kayal11} for details).

At step 1, the Hessian determinant can be factorized by Kaltofen's algorithm \cite{Kaltofen89} for the factorization of arithmetic circuits as suggested in \cite{kayal11}, 
or by the black box factorization algorithm of Kaltofen  and Trager \cite{KalTra90}.
These two algorithms assume access to an algorithm for the factorization of univariate polynomials 
(one of the algorithms in~\cite{koiran2018orbits} reduces instead to the closely related task of matrix diagonalization).
For $K=\mathbb{K}=\cc$ one can just assume the ability to factor univariate polynomials as part of our computation model.
This yields a polynomial time algorithm, which is clearly not designed to run on a Turing machine.
Another option is to take $K=\mathbb{K}=\qq$, and we obtain a polynomial time algorithm for the Turing machine model.

Assume now that $K=\qq$, $\mathbb{K}=\cc$ and that we wish to design again an algorithm for the Turing machine model.
As mentioned in Section \ref{subsec:equiv}, a natural approach would be to factor $H_f$ symbolically at step 1, i.e., to construct 
an extension $K'$ of $\qq$ of finite degree where we can find the coefficients of the linear forms $l_i$.
The linear algebra computations of step 2 would then be carried out symbolically in $K'$.
It is not clear that this approach yields a polynomial time algorithm even for $d=3$ because these computations could possibly take place 
in an extension of exponential degree (recall indeed that the splitting field of a univariate polynomial of degree $r$ may be of degree as high as $r!$).
We provide polynomial time algorithms for this problem (and for $\mathbb{K}=\rr$) in Sections~\ref{sec:random} and~\ref{sec:deter}.
In order to stay closer to Kayal's original algorithm, 
{  a plausible approach would be to stop his algorithm at step 1.
 Note indeed that step 3 is not necessary for $\mathbb{K}=\cc$, and it is not immediately clear whether step 2 is necessary. Namely, it is not obvious whether there are polynomials
that pass the factorization test of step 1 but fail at step 2. 
This led us to the following question: 

\begin{question} \label{question:hessian}
  Let $f \in \cc[x_1,\ldots,x_n]$ be a homogeneous polynomial of degree $d \geq 3$. If the Hessian determinant of $f$ is equal to $(x_1x_2\cdots x_n)^{d-2}$,
must $f$ be of the form $f(x_1,\ldots,x_n)=\alpha_1x_1^d+\cdots+\alpha_nx_n^d$?
\end{question}
A positive answer would yield a polynomial time decision algorithm
for the equivalence problem since the {\em existence} of a suitable factorization at step~1 can be decided
in polynomial time~\cite{koiran2018orbits}.
Representation of polynomials by Hessian determinants has 
proved
to be a delicate topic: see \cite{gondim15} for a famous mistake by Hesse about his eponymous determinant. {  Hesse's mistake was about polynomials with vanishing Hessian, a topic that remains of interest to this day \cite{huang19}}.
One of the authors of \cite{huang19} came accross Question~\ref{question:hessian} in an earlier version of the present paper, and managed to obtain a negative answer for many cases of interest~\cite{ventura21}:
$ n \geq 2$ and $d \geq 4$ even, or $ n \geq 3$ and $3 \leq d \leq n$, or $ n \geq 3$ and and $d=kn$
for some $k>1$. 
The approach pursued in our paper, based on simultaneous diagonalization by congruence, therefore remains the only way of testing equivalence to a sum of cubes over~$\cc$ in polynomial time. 
We now present this approach in detail.   }

\subsection{Equivalence by tensor decomposition} \label{sec:tensor}

In the remainder of this section we explain our approach to the equivalence problem. Like in most of the paper, we work in a field $\kk$ which
  is either the field of real or complex numbers.
Recall that we can associate to a symmetric tensor $T$ of order 3 the homogeneous polynomial $f(x_1,\ldots,x_n)=\sum_{i,j,k=1}^n T_{ijk} x_i x_j x_k$.
This correspondence is bijective, and the symmetric tensor associated
to a homogeneous polynomial $f$ can be obtained from the relation:
\begin{equation}\label{partial3}
\displaystyle \frac{\partial^3 f}{\partial x_i \partial x_j \partial x_k}  = 6 T_{ijk}.
\end{equation}
The $i$-th slice of $T$ is the symmetric matrix $T_i$ with entries $(T_i)_{jk} = T_{ijk}$.
By abuse of language, we will also say that $T_i$ is the $i$-th slice of $f$.
Note that~(\ref{partial3}) is the analogue of the relation
$$\frac{\partial^2 q}{\partial x_i \partial x_j}  = 2 Q_{ij}$$
which connects the entries of a symmetric matrix $Q$ to the  partial derivatives of the quadratic from $q(x)=x^TQx$.
Comparing these two equations shows that the matrix of the quadratic form $\partial f / \partial x_k$ is equal to $3T_k$.

\begin{remark} \label{rem:diag}
The slices of a polynomial of the form 
\begin{equation} \label{eq:diagtensor}
g(x_1,\ldots,x_n)=\alpha_1 x_1^3+\ldots+\alpha_n x_n^3
\end{equation}
are the diagonal matrices $\diag(\alpha_1,0,\ldots,0),\ldots,\diag(0,\ldots,0,\alpha_n)$.
Conversely, if all the slices of a degree 3 homogeneous polynomial $g$ are diagonal then $g$ must be of the above form (in particular, such a $g$ is equivalent to a sum of $n$ cubes 
iff the coefficients $\alpha_i$ are all nonzero;
this follows from the fact that
for $\kk \in \{\rr,\cc\}$, any element of $\kk$ has a cube root in $\kk$). Indeed, the presence of any other monomial in $g$
would yield an off-diagonal term in some slice: for the monomial $m=x_i^2x_j$ with $i \neq j$ we have $\partial m / \partial x_i = 2x_i x_j$ and for
$m=x_i x_j x_k$ with all indices distinct we have
$\partial m / \partial x_i = x_j x_k$.
\end{remark}

In light of Definition \ref{def:equiv}, it is important to understand how slices behave under a linear change of variables. This was done  for symmetric and ordinary tensors in \cite[Section 2.1 and Proposition 48]{koiran2019ortho}. In particular, for symmetric tensors  the following result can be obtained from (\ref{partial3}):
\begin{proposition} \label{prop:slices}
Let $g$ be a degree 3 form with slices $S_1,\ldots,S_n$ and let $f(x)=g(Ax)$. The slices $T_1,\ldots,T_n$ of $f$ are given by the formula:
$T_k=A^TD_kA$ where $D_k=\sum_{i=1}^n a_{ik} S_i$ and the $a_{ik}$ are the entries of $A$. In particular, if $g$ is as in (\ref{eq:diagtensor}) 
we have $D_k = \diag(\alpha_1 a_{1k},\ldots,\alpha_n a_{nk})$.
\end{proposition}
A similar property appears in the analysis of Jennrich's algorithm~\cite[Lemma~3.3.3]{moitra18}.
The action on slices given by the formula $T_k=A^TD_kA$ in this proposition seems at least superficially related to the action on tuples of symmetric (and antisymmetric) matrices 
studied by Ivanyos and Qiao \cite{ivanyos19}. They consider an action of $GL_n$ sending a tuple $(S_1,\ldots,S_m)$ to the tuple $(T_1,\ldots,T_m)$ where $T_i=A^TS_iA$. 
Two tuples are said to be isometric if there exists an invertible matrix $A$ realizing this transformation.
Some of the main differences with our setting are:
\begin{itemize}
\item[(i)] The number of elements in our matrix tuples is the same as the dimension $n$ of the matrices, but in their setting $m$ and $n$ are unrelated.
\item[(ii)] The matrices in our tuples must come from a symmetric tensor but they allow arbitrary tuples of symmetric matrices.
\item[(iii)] They act independently on each component of a matrix tuple, whereas we "mix" components with the transformation $D_k=\sum_{i=1}^n a_{ik} S_i$. In spite of this difference, the actions on the space of matrices spanned by the tuple's components are the same, see Lemma \ref{lem:slicespan} below.
\end{itemize}
Also we note that their algorithm for isometry testing is not algebraic since it requires the construction of field extensions as explained e.g. in the paragraph on the representation of fields and field extensions in  \cite{ivanyos19}.\footnote{An algebraic algorithm does not require the construction of field extensions since by definition all operations take place in the ground field. In \cite{ivanyos19} they only need to construct extensions of polynomially bounded degree. As explained in Section~\ref{subsec:kayal}, this is not clear for Kayal's algorithm.}

\begin{lemma} \label{lem:slicespan}
Let $f(x_1,\ldots,x_n)$ and $g(x_1,\ldots,x_n)$ be two forms of degree 3 such that $f(x)=g(Ax)$ for some nonsingular matrix $A$. 
\begin{itemize} 
\item[(i)] If $\cal U$ and $\cal V$ denote the
subspaces of $M_n(\kk)$ spanned respectively by the slices of $f$ and $g$, we have ${\cal U} = A^T {\cal V} A$. 

\item[(ii)] In particular, for $g=P_3$ the subspace $\cal V$ is the space of diagonal matrices and  ${\cal U}$   is a nonsingular subspace, i.e., it is not made of singular matrices only.
\end{itemize}
\end{lemma}

\begin{proof}
Proposition \ref{prop:slices} shows that ${\cal U} \subseteq A^T {\cal V} A$. 
Since $g(x) = f(A^{-1}x)$, the same argument shows that  ${\cal V} \subseteq A^{-T}{\cal U} A^{-1}$. The inclusion ${\cal U} \subseteq A^T {\cal V} A$ therefore cannot be strict.
The second part of the lemma follows immediately from the first
and from Remark~\ref{rem:diag}.
\end{proof}
    For the next theorem, we recall from the beginning of
      Section~\ref{sec:tensor} that one
may take either $\kk=\rr$ or $\kk = \cc$.
\begin{theorem} \label{th:carsumofcubes}
A degree 3 form $f \in \kk[X_1,\ldots,X_n]$ is equivalent to a sum of~$n$ cubes if and only if its slices $T_1,\ldots,T_n$ span a nonsingular matrix space and 
 the slices are simultaneously diagonalizable by congruence, i.e., there exists an invertible matrix $Q \in M_n(\kk)$ such that the $n$ matrices $Q^T T_i Q$ are diagonal.
\end{theorem}
\begin{proof}
Let $\cal U$ be the space spanned by  $T_1,\ldots,T_n$.
If $f$ is equivalent to a sum of $n$ cubes, Proposition \ref{prop:slices} shows that the slices of $f$ are simultaneously diagonalizable by congruence 
and Lemma \ref{lem:slicespan} shows that $\cal U$ is  nonsingular.

Let us show the converse. Since the slices are simultaneously diagonalizable by congruence, there are   diagonal matrices $\Lambda_k$ and a  nonsingular matrix 
$R \in M_n({ \kk})$ such that 
$T_k = R \Lambda_k R^T$ for all $k=1,\ldots,n$. 
Let $g(x)=f(R^{-T}x)$. By Proposition \ref{prop:slices} the slices of $g$ are linear combinations of the $\Lambda_k$, i.e., they are all diagonal.
By Remark \ref{rem:diag}, $g$ must be as in (\ref{eq:diagtensor}). 
It therefore remains to show that the coefficients $\alpha_i$ are all nonzero. This must be the case due to the hypothesis on $\cal U$.
Indeed, this hypothesis implies that the matrix space $\cal V$ spanned by the slices of $g$ is nonsingular (apply again Lemma \ref{lem:slicespan}, this time in the other direction).
But if some $\alpha_i$ vanishes,  $\cal V$  is included in the space of diagonal matrices with a 0 in the $i$-th diagonal entry.
\end{proof}

\begin{corollary} \label{cor:sumofcubes}
Let $f$ be a degree 3 form with slices $T_1,\ldots,T_n$ and assume that $T_1$ is nonsingular. Then $f$ is equivalent to a sum of $n$ cubes if and only if the $n-1$ matrices 
 $T_1^{-1}T_k$ ($k=2,\ldots,n$) commute and are diagonalizable over $\kk$.
 \end{corollary}
 \begin{proof}
This follows from Theorem \ref{th:carsumofcubes}  as well as Theorem \ref{th:simdiag} for $\kk=\cc$ and Theorem \ref{th:realsimdiag} for $\kk=\rr$.
\end{proof}

We conclude this section with an alternative characterization of equivalence to a sum of cubes for the field of real numbers.
\begin{theorem} \label{th:pdsumofcubes}
Let $f$ be a real form of degree 3 and let $\cal V$ be the subspace of~$M_n(\rr)$ spanned by the slices of $f$.
The 3 following properties are equivalent:
\begin{itemize}
\item[(i)]  $f$ is equivalent as a real polynomial to a sum of $n$  cubes.
\item[(ii)]  There exist two invertible matrices $A,B \in \cal V$ such that $A^{-1} \cal V$ is a commuting subspace and $B$ is positive definite.
\item[(iii)] $\cal V$ contains a positive definite matrix, and $A^{-1} \cal V$ is a commuting subspace  for any invertible matrix $A \in \cal V$.
\end{itemize}
\end{theorem}
\begin{proof}
Suppose that $f$ is  equivalent to a sum of $n$ cubes, i.e., $f(x)=P_3(Qx)$ where $Q \in M_n(\rr)$ is invertible. 
We have seen in Lemma \ref{lem:slicespan} that the slices of $P_3$ span the space $\cal D$ of diagonal matrices, and that those of $f$ span $Q^T {\cal D} Q$.
The latter span contains the positive definite matrix $B=Q^T Q$.
Moreover, according to Proposition \ref{prop:slices} the slices of $f$ are simultaneously diagonalizable by congruence. 
By Theorem \ref{th:realdissociate}, $A^{-1} \cal V$ is a commuting subspace  for any invertible matrix $A \in \cal V$.
Hence we have shown that (i) implies (iii).
That (iii) implies (ii) is clear since $\cal V$ is nonsingular (by hypothesis, it contains a positive definite matrix).
Finally, let us show that (ii) implies (i). By hypothesis, $\cal V$ contains a positive definite matrix $B$ hence we can apply Theorem \ref{th:pd}. It follows that the slices are simultaneously diagonalizable by congruence. By Theorem \ref{th:carsumofcubes}, $f$ must be equivalent to a sum of $n$ cubes.
\end{proof}
Compared to Theorem \ref{th:carsumofcubes} or Corollary \ref{cor:sumofcubes}, Theorem \ref{th:pdsumofcubes} does not involve any diagonalizability test.
One can check that $\cal V$ contains a positive definite matrix using semi-definite programming. Unfortunately, no efficient algebraic algorithm is known for semi-definite programming (famously, this is already an open problem for linear programming). For this reason, the equivalence algorithms
of this paper will be based  on Corollary \ref{cor:sumofcubes} rather than Theorem \ref{th:pdsumofcubes}.

\section{Randomized equivalence algorithm} \label{sec:random}

As a test for equivalence to a sum of $n$ cubes, Corollary \ref{cor:sumofcubes} is not quite satisfactory due to the hypothesis on $T_1$ (note indeed that this hypothesis is not even satisfied
by $f=P_3$). This restriction can be overcome by performing a random change of variables before applying Corollary \ref{cor:sumofcubes}.
This yields the following simple randomized algorithm with one-sided error. The input is a degree 3 form $f \in \kk[x_1,\ldots,x_n]$.
We recall from Section~\ref{sec:tensor}
that $\kk = \rr$ or $\kk = \cc$ (except in Proposition~\ref{diagmat} where any field
of characteristic 0 is allowed).
\begin{enumerate}
\item Pick a random matrix $R \in M_n(\kk)$ and set $h(x)=f(Rx)$.
\item Let $T_1,\ldots,T_n$ be the slices of $h$. If $T_1$ is singular, reject. Otherwise, compute $T'_1=T_1^{-1}$.
\item If  the matrices $T'_1T_k$ commute and are all diagonalizable over $\kk$, accept. Otherwise, reject.
\end{enumerate}
Before proving the correctness of this algorithm, we explain how the diagonalizability test at step 3 can be implemented efficiently with an algebraic algorithm.
This can be done thanks to the following classical result from linear algebra
(see e.g.~\cite[Corollary~3.3.8]{horn13} for the case $\mathbb{K} = \cc$).
\begin{proposition} \label{diagmat}
  Let $\mathbb{K}$ be a field of characteristic 0 and let 
  $\chi_M$ be the characteristic polynomial of a matrix $M \in M_n(\mathbb{K})$.
  Let $P_M = \chi_M / \mathrm{gcd}(\chi_M,\chi_M')$ be the squarefree part of
  $\chi_M$. The matrix $M$ is diagonalizable over $\overline{\mathbb{K}}$ iff
  $P_M(M)=0$.
    Moreover, in this case $M$ is diagonalizable over $\mathbb{K}$ iff all the
  roots of $P_M$ lie in $\mathbb{K}$.
\end{proposition}
Over the field of complex numbers it therefore suffices to check that $P_M(M)=0$. Over $\rr$, we need to check additionally that all the roots of $P_M$ are real.
This can be done for instance with the help of Sturm sequences, which can be used to compute the number of roots of a real polynomial on any real (possibly unbounded) interval. Alternatively, the number of real roots of a real polynomial can be obtained through Hurwitz determinants 
\cite[Corollary 10.6.12]{rahman2002}, and is given by the signature of the Hermite quadratic form \cite[Theorem 4.48]{Basu06}.
The arithmetic cost of these methods is polynomially bounded, and they can also be implemented to run in polynomial time in the bit model.\footnote{For Sturm sequences
this is not obvious because in a naive implementation, the bit size of the numbers involved may grow exponentially. 
There is however an efficient implementation based on subresultants \cite{Basu06}. The same issue of coefficient growth already occurs in the computation of the gcd of two polynomials, and can also be solved with subresultants.}
\begin{theorem} \label{th:random}
If an input $f \in \kk[X_1,\ldots,X_n]$ is accepted by a run of the above randomized algorithm then $f$ must be equivalent to a sum of $n$ cubes.

Conversely, if  $f$ is equivalent to a sum of $n$ cubes then $f$ will be accepted with high probability over the choice of the random matrix $R$ at step 1.
More precisely, if the entries $r_{ij}$ are chosen independently at random from a finite set $S$ the input will be accepted with probability at least $1-2n/|S|$.
\end{theorem}

\begin{proof}
Assume that $f$ is accepted for some choice of $R \in M_n(\cc)$. Since $T_1$ is invertible, it follows from Proposition \ref{prop:slices} that $R$ must be invertible as well.
Moreover, $h$ must be equivalent to a sum of $n$ cubes by Corollary \ref{cor:sumofcubes}. The same is true of $f$ since $f(x)=h(R^{-1}x)$.

For the converse, assume that $f$ is equivalent to a sum of $n$ cubes. We can obtain the slices $T_k$ of $h$ from the slices $S_k$ of $f$ by Proposition \ref{prop:slices}, 
namely, we have $T_k=R^TD_kR$ where $D_k=\sum_{i=1}^n r_{ik} S_i$ and the $r_{ik}$ are the entries of~$R$.
Therefore $T_1$ is invertible iff $R$ and $D_1$ are invertible. By Lemma \ref{lem:slicespan}.(ii)  there is a way to choose the entries $r_{i1}$ so that $D_1$ 
is invertible. In fact, $D_1$ will be invertible for most choices of these entries. This follows from the fact that as a polynomial in the entries $r_{11},\ldots, r_{n1}$, 
$\det(D_1)$ is not identically zero. Therefore, by the Schwartz-Zippel lemma $D_1$ will fail to be invertible with probability at most $n/|S|$. Likewise, $R$ will fail to be invertible 
with probability at most  $n/|S|$ and the result follows from the union bound.
\end{proof}

\begin{remark}
In Theorem \ref{th:random} and in the corresponding algorithm, we can reduce the amount of randomness by picking random matrices $R$ of the following special form:
$R$ is lower triangular with 1's on the diagonal (except possibly for $r_{11}$), 
$r_{11},\ldots,r_{n1}$ are drawn independently and uniformly from $S$, and all the other entries are set to 0.
The same analysis as before shows that $D_1$ will fail to be invertible with probability at most  $n/|S|$.
Moreover, $R$ will fail to be invertible with probability at most $1/|S|$ since $\det(R)=r_{11}$. By the union bound, $f$ will be accepted with probability at least $1-(n+1)/|S|$.
\end{remark}
We will use a similar construction in the deterministic algorithm of the next section.

\section{Deterministic equivalence algorithm} \label{sec:deter}

In the analysis of our randomized algorithm we have invoked the Schwartz-Zippel lemma to argue that a polynomial of the form 
$$H(r_1,\ldots,r_n)=\det(r_1S_1+\ldots+r_nS_n)$$ does not vanish for most of the random choices $r_1,\ldots,r_n$ (recall from the proof of Theorem~\ref{th:random} that $S_1,\ldots,S_n$ denoted the slices of $f$).
In this section we will obtain obtain our deterministic equivalence algorithm by derandomizing this step. Namely, we will use the fact that we are not trying to solve an
arbitrary instance of symbolic determinant identity testing: as it turns out, the polynomial $H$ can be factored as a product of linear forms.
This fact was already at the heart of Kayal's equivalence algorithm. Indeed, his algorithm is based on the factorization of the Hessian determinant of $f$ 
\cite[Lemma 5.2]{kayal11}
and as pointed out in \cite{koiran2019ortho}, the symbolic matrix $r_1S_1+\ldots+r_nS_n$ is a constant multiple of the Hessian.
The point where we depart form Kayal's algorithm is that we do not explicitly factor $H$ as a product of linear forms (recall indeed that this is not an algebraic step).
Instead, we will use the {\em existence} of such a factorization to find deterministically a point where $H$ does not vanish. 
We can then conclude as in the previous section.

First, we formally state this property of $H$ as a lemma and for the sake of completeness we show that it follows from Proposition \ref{prop:slices} (one can also make this argument in the opposite direction, see Section 2.1 of \cite{koiran2019ortho} for details).
\begin{lemma} \label{lem:product}
  Let $f$ be a degree 3 form with slices $S_1,\ldots,S_n$
  and let $H(x_1,\ldots,x_n)= \det(x_1S_1+\ldots+x_nS_n)$.
  If $f$ is equivalent to a sum of~$n$ cubes then $H$ is not identically 0 and can be factored as a product of $n$ linear forms.
\end{lemma}
\begin{proof}
Let $A$ be the invertible matrix such that $f(x)=P_3(Ax)$. 
By Proposition \ref{prop:slices}, $H(x)=(\det A)^2 \det D(x)$ where $$D(x)=\sum_{k=1}^n x_k D_k = \diag(a_{11}x_1+\cdots+a_{1n}x_n,\ldots,a_{n1}x_1+\cdots+a_{nn}x_n).$$
This gives the required factorization. In particular, $H$ is nonzero since $A$ is invertible. 
\end{proof}
The non vanishing of $H$ means that the slices span a nonsingular matrix space. We have given in Theorem \ref{th:carsumofcubes} a slightly different proof of the fact that this space is indeed nonsingular when $f$ is equivalent to a sum of $n$ cubes.
By Lemma \ref{lem:product}, the zero set of $H$ is a union of $n$ hyperplanes. We can avoid the union of any finite number of hyperplanes by a standard construction
involving the {\em moment curve} $\gamma(t)=(1,t,t^2,\ldots,t^{n-1})$.
\begin{lemma} \label{lem:moment}
Let $M \subseteq \cc^n$ be a set of  $(n-1)p+1$ points on the moment curve. For any set of $p$ hyperplanes $H_1,\ldots,H_p \subseteq \cc^n$ there is at least one point
of $M$ which does not belong to any of the $H_i$.
 \end{lemma}
 \begin{proof}
 Let $l_i(x_1,\ldots,x_n)=0$ be the equation of $H_i$. The moment curve has at most $n-1$ intersections with $H_i$ since $l_i(1,t,t^2,\ldots,t^{n-1})$ is a nonzero
 polynomial of degree $n-1$. For the $p$ hyperplanes we therefore have a grand total of $p(n-1)$ intersection points at most.
 \end{proof}
 The size of $M$ in this lemma is the smallest that can be achieved in such a blackbox construction. Indeed, for any set of $(n-1)p$ points one can always find 
 a set of $p$ hyperplanes which covers them all.
 
We can now describe our deterministic algorithm. As in Section \ref{sec:random}  the input is a degree 3 form $f(x_1,\ldots,x_n)$ with slices $S_1,\ldots,S_n$.
\begin{enumerate}
\item Pick an arbitrary set $M$ of $n(n-1)+1$ points on the moment curve.
\item Enumerate the elements of $M$ to find a point $r=(1,r_2,\ldots,r_n) \in M$ such that the matrix $D_1=S_1+r_2S_2\ldots+r_nS_n$ is invertible. 
If there is no such point, reject.
\item Construct the following matrix $R \in M_n(\kk)$: $R$ is lower  triangular with 1's on the diagonal,
$r_{21}=r_2,\ldots,r_{n1}=r_n$  and all the other entries are set to 0.
\item Compute $h(x)=f(Rx)$, the slices $T_1,\ldots,T_n$ of $h$ and $T'_1=T_1^{-1}$.
\item  If  the matrices $T'_1T_k$ commute and are all diagonalizable, accept. Otherwise, reject.
\end{enumerate}

\begin{theorem}
A degree 3 form $f(x_1,\ldots,x_n)$ is accepted by the above algorithm if and only if $f$ is equivalent to a sum of $n$ cubes.
\end{theorem}
\begin{proof}
As a preliminary observation, we note that if the algorithm reaches step 4 the matrix $T'_1$ is well-defined since $T_1$ is invertible. Indeed, we have seen in the proof of Theorem \ref{th:random} that $T_1=R^TD_1R$; moreover, $D_1$ is invertible since the algorithm has not failed at step 2 and $R$ is clearly invertible as well.

Suppose now that an input $f(x_1,\ldots,x_n)$ is accepted by the algorithm. 
The same argument as in the proof of Theorem \ref{th:random} shows that $f$ is equivalent to a sum of $n$ cubes. Namely,  $h$ must be equivalent to a sum of $n$ cubes by Corollary \ref{cor:sumofcubes}. The same is true of $f$ since $R$ is an invertible matrix.

For the converse, suppose that  an input $f(x_1,\ldots,x_n)$ is equivalent to a sum of $n$ cubes. By Lemma \ref{lem:product} and Lemma \ref{lem:moment}, there exists
a point $r \in M$ where the polynomial $H(r)=\det(r_1S_1+\ldots+r_nS_n)$ does not vanish. As a result, the algorithm will not reject at step 2.
Since the matrix $R$ constructed at step 3 is invertible, the polynomial $h$ at step 4 is equivalent to a sum of $n$ cubes and the algorithm will accept at step 5 by 
Corollary \ref{cor:sumofcubes}.
\end{proof}

\begin{remark}
Some of the results in the paper by Ivanyos and { Qiao} \cite{ivanyos19} mentioned after Proposition \ref{prop:slices} are motivated by an application to symbolic determinant identity testing (SDIT). 
In our setting we only need to consider very simple determinants (as explained at the beginning of this section, they factor as a product of linear forms).
As a result we can use the simple black box solution provided by Lemma \ref{lem:moment}. More connections between group actions and SDIT can be found in \cite{garg16,ivanyos17a,ivanyos17b}.
\end{remark}

\section{Polynomial Identity Testing} \label{sec:PIT}

It is a basic fact that  black box PIT for a class of polynomials $\cal C$ is equivalent to constructing a hitting set for $\cal C$, i.e., a set of points $H$ such that every polynomial
in $ \cal C$ which vanishes on all points of $H$ must vanish identically. Indeed, from a hitting set we obtain a black box PIT algorithm by querying the input polynomial $f$
at all points of $H$. Conversely, for any black box PIT algorithm the set of points queried on the input $f \equiv 0$ must form a hitting set. Note that the validity of this
simple argument depends on the hypothesis that $0 \in \cal C$ (otherwise we can declare that $f {\not \equiv} 0$ without making any query).

In this section we first consider the following scenario. An algorithm is provided with black box access to a polynomial $f$ that is either identically 0 or equivalent to $P_d$,
and must decide in which of these two categories its input falls (note that these are indeed two disjoint cases).
This is equivalent to constructing a hitting set for the equivalence class of $P_d$, a task that we carry out in Section \ref{sec:hitting_equiv}.
Then in Section~\ref{sec:fewer}
  we generalize this hitting set construction to a
  larger class of polynomials, namely, those that can be written as sums
  of $d$-th powers of linearly independent linear forms.

\subsection{A hitting set for the equivalence class of $P_d$} \label{sec:hitting_equiv}

Here we construct a polynomial size hitting set for the set of polynomials $f \in \cc[X_1,\ldots,X_n]$ that are equivalent to $P_d$.
{  Let $S = \{s_0,s_1,\ldots,s_d\}$ be a set of $d+1$ complex numbers with $s_0=0$.
We denote by $S_i$ the set of points $(x_1,\ldots,x_n) \in \cc^n$ with $x_i \in S$ and all other coordinates equal to 0. 
Pick an arbitrary nonzero point $p \in \cc^n$, and form the set $G_p = p+\bigcup_{i=1}^n S_{i},$ of size $nd+1$.

\begin{proposition} \label{prop:hitting_equiv}
For any $d \geq 2$ and any nonzero point $p \in \cc^n$, $G_p$ is a hitting set for the set of polynomials $f \in \cc[X_1,\ldots,X_n]$ that are equivalent to $P_d$.
\end{proposition}
\begin{proof}
Let $f(x)=P_d(Ax)$ where $A$ is an invertible matrix. The proof is based on the fact that the gradient 
$\nabla f = (\partial f / \partial x_1,\ldots,\partial f / \partial x_n)$ does not vanish anywhere except at $x=0$.
Indeed, this property clearly holds true for $P_d$, and is preserved by an invertible change of variables
since 
\begin{equation} \label{eq:grad}
 \nabla f (x) = A^T \nabla P_d(Ax)
\end{equation}
by the chain rule.
In particular, $\nabla f (p) \neq 0$ since $p \neq 0$. 
This implies that $f$ does not vanish on all of $G_p$.
Indeed, if $f$ vanishes on $p+S_i$ then $f$ vanishes on the whole line going through $p+S_i$,
hence $\frac{\partial f}{ \partial x_i} (p)=0$.

\end{proof}}

Let $K$ be a subfield of $\cc$. 
The same construction yields a hitting set for the set of polynomials in $K[x_1,\ldots,x_n]$ that are equivalent to a polynomial of the form $\sum_{i=1}^n a_i x_i^d$ 
where $a_i \in K \setminus \{0\}$.
Such polynomials are indeed equivalent to $P_d$ as complex polynomials.

\subsection{Fewer powers} \label{sec:fewer}

We  now give a black box PIT algorithm 
for a bigger class of polynomials, namely, those that can be written as sums of  $d$-th powers  of linearly independent  linear forms.
These polynomials therefore  admit decompositions as in~(\ref{eq:linind})
where the forms $l_i$ are linearly independent and the number $r \leq n$ of forms in the decomposition is unknown. 
{  We will generalize the approach from Section \ref{sec:hitting_equiv}.
In particular, we will see that the gradient of $f$ does not vanish outside of a certain (unknown) linear subspace $V$.
We can find deterministically a point outside of $V$ with the help of the moment curve like in Section~\ref{sec:deter}.
This leads to the construction of the following hitting set: for an arbitrary set $M$ of $n$ points on the moment curve,
we construct  the union of the $G_p$'s as $p$ ranges over $M$ (recall that $G_p$ is the hitting set of Section~\ref{sec:hitting_equiv}). This set $G \subseteq \cc^n$ is of size $n(nd+1)$. 

\begin{theorem} \label{th:hitfew}
 For any $d \geq 2$, $G$ is a hitting set for the set of polynomials $f \in \cc[X_1,\ldots,X_n]$ that 
 can be written as sums of $d$-th powers of linearly independent linear forms. 
  \end{theorem}
\begin{proof}
Suppose that $f$ can be written as a sum of $r$ $d$-th powers for some $r \geq 1$.
 We have $f(x)=P_{d,r}(Ax)$ where $A$ is an invertible matrix and 
 $$P_{d,r}(x_1,\ldots,x_n)=x_1^d+\cdots+x_r^d.$$
We will use the fact that the gradient of $f$ vanishes only on a (proper) linear subspace $V$ of dimension $n-r$.
This property clearly holds true for $P_{d,r}$, and it is preserved for $f$ since we now have
$$\nabla f (x) = A^T \nabla P_{d,r}(Ax)$$
instead of~(\ref{eq:grad}).
By Lemma~\ref{lem:moment}, $M$ contains at least one point $p$ lying outside of $V$. 
At this point $\nabla f (p) \neq 0$, and the same argument as in the proof of Proposition~\ref{prop:hitting_equiv} 
shows that $f$ does not vanish on $G_p$.
  \end{proof}}

\section{Linear dependencies, essential variables and Lie algebras}
\label{sec:essential}

In this section we build on the results from Section~\ref{sec:PIT}
  to derandomize several algorithms from~\cite{kayal11,kayal2012affine}.
  We begin in Section~\ref{sec:depend}
  with the computation of linear dependencies between polynomials.
  Then we give applications to the minimization of the number of variables in sums of powers of linear forms
  (in Section~\ref{sec:min}), and to the computation of Lie algebras of products of linear forms (in Section~\ref{sec:Lie}).
  This leads to the derandomization of a factorization algorithm from~\cite{koiran2018orbits} and of the equivalence
  algorithm by Kayal~\cite{kayal11} described in Section~\ref{subsec:kayal}.

\subsection{From black box PIT to linear dependencies} \label{sec:depend}

We first recall from~\cite{kayal11}
the notion of {\em linear dependencies among polynomials}.
It has found applications to the elimination of redundant
variables~\cite{kayal11},
the computation of the Lie algebra of a polynomial~\cite{kayal2012affine},
the reconstruction of random arithmetic formulas~\cite{gupta14},
full rank algebraic programs~\cite{kayal18}
and non-degenerate depth 3 circuits~\cite{kayal19}.
\begin{definition}
  Let ${\bf f} = (f_1,\ldots,f_m)$ be a tuple of $m$ polynomials of
  $K[X_1,\ldots,X_n]$. The space of linear dependencies of ${\bf f}$,
  denoted ${\bf f}^{\perp}$, is the space of all vectors
  $v=(v_1,\ldots,v_m) \in K^m$ such that $v_1f_1+\cdots+v_mf_m$ is identically 0.
\end{definition}
As a computational problem, the POLYDEP problem consists of finding
a basis of ${\bf f}^{\perp}$ for a tuple ${\bf f}$ given as input.
If the $f_i$ are verbosely given as sum of monomials, this is a simple problem
of linear algebra. The problem becomes more interesting if the $f_i$
are given by arithmetic circuits or black boxes.
In Section~\ref{sec:depend} we present a simple and general relation between this problem
and black box PIT.

A natural approach to POLYDEP consists of evaluating the $f_j$ at certain
points $a_1,\ldots,a_k$ of $K^n$ to form a $k \times m$ matrix $M$ with the
$f_j(a_i)$ as entries. Note that ${\bf f}^{\perp} \subseteq \ker(M)$ for
any choice of the evaluation points. We would like this inclusion to be
an equality since this will allow to easily compute a basis of ${\bf f}^{\perp}$.
This motivates the following definition.
\begin{definition}
  The points $a_1,\ldots,a_k$ form a
  {\em hitting set for the linear dependencies of ${\bf f}$} if the above
  matrix $M=(f_j(a_i))_{1 \leq i \leq k, 1\leq j \leq m}$ satisfies
  ${\bf f}^{\perp} = \ker(M)$.
\end{definition}
Kayal~\cite{kayal11}
showed (without using explicitly this terminology) that if $k=m$ and
the $a_i$ are chosen at random, a
{hitting set for the linear dependencies of ${\bf f}$}
will be obtained with high probability.

Here we point out that constructing deterministically
a hitting set for the linear dependencies of ${\bf f}$ is {\em  equivalent}
  to solving black box PIT for the family of polynomials in
  $\mathrm{Span}({\bf f})$
  (the space of linear combinations of $f_1,\ldots,f_m$):
  \begin{proposition} \label{prop:hitequiv}
    Let ${\bf f} = (f_1,\ldots,f_m)$ be a tuple of $m$ polynomials of
  $K[X_1,\ldots,X_n]$. For any tuple  $(a_1,\ldots,a_k)$ of $k$ points of  $K^n$,
    the two following properties are equivalent:
    \begin{itemize}
      \item[(i)] The points $a_1,\ldots,a_k$ form a
        {hitting set for the linear dependencies of~${\bf f}$}.
        \item[(ii)] They form a hitting set for $\mathrm{Span}({\bf f})$.
     \end{itemize}
  \end{proposition}
  \begin{proof}
    This is    immediate from the definitions.
    Suppose indeed that (i) holds, and that some polynomial
    $f=v_1f_1+\ldots+v_mf_m$ of $\mathrm{Span}({\bf f})$ vanishes
    at all of the $a_i$. This means that $v \in \ker M$,
    hence $v \in {\bf f}^{\perp}$ by (i). We conclude that $f$ is identically 0
    and (ii) holds.

    To prove the converse we can take the same steps in reverse.
    Suppose that (ii) holds and that $v \in \ker M$. This means that
    $f=v_1f_1+\ldots+v_mf_m$ vanishes at all the $a_i$,
    hence $f$ is identically 0 by (ii). We have shown that
    $v \in {\bf f}^{\perp}$, i.e., ${\bf f}^{\perp} = \ker M$.
  \end{proof}
  In Section \ref{sec:min} we will use this observation and the black box
  PIT algorithm of Section \ref{sec:fewer} to minimize the number of  variables
  in sums of powers of linearly independent linear forms.
  In Section \ref{sec:Lie} we give an application to the computation of Lie algebras
  and factorization into products of linear forms.

   \subsection{Minimizing variables} \label{sec:min}

  We first recall the notion of {\em redundant} and {\em essential} variables
  studied by Carlini~\cite{Carlini06} and Kayal~\cite{kayal11}.
  \begin{definition}
    A variable $x_i$ in a polynomial $f(x_1,\ldots,x_n)$ is redundant if~$f$
    does not depend on $x_i$, i.e.,
    $x_i$ does not appear in any monomial of $f$.

    We say that $f$ has $t$ essential variables if $t$ is the smallest
    number for which there is an invertible
    matrix of size $n$ such that $f(Ax)$ depends on $t$ variables only.
  \end{definition}
  A randomized algorithm for minimizing the number of variables
  is given in~\cite[Theorem 4.1]{kayal11}.
  More precisely, if the input $f$ has $t$ essential
  variables the algorithm finds (with high probability) an invertible matrix~$A$
  such that $f(Ax)$ depends on its first $t$ variables only.
  It is based on the observation from~\cite{Carlini06,kayal11} that
  $t= n - \dim (\partial f)^{\perp} = \dim (\partial f)$
  where $\partial f$ denotes the tuple
  of~$n$ partial derivatives $\partial f / \partial x_i$
  (and  $\dim (\partial f)$ denotes the dimension of the spanned subspace).
  As recalled in Section~\ref{sec:depend}, a basis of the space
  of linear dependencies $ (\partial f)^{\perp}$
  can be found by a randomized algorithm
  from~\cite{kayal11}. Moreover, a suitable invertible matrix $A$
  is easily found from such a basis by completing it into a basis
  of the whole space $K^n$ (see appendix~B of~\cite{kayal11} for details).
  \begin{example}\label{ex:essvar}
    If $f$ can be written as a sum of $r$ powers
    of linearly independent linear forms then the number of essential variables
    of $f$ is equal to~$r$.
    This is clear for $f(x_1,\ldots,x_n)=x_1^d+\cdots+x_r^d$ since $\partial f$
    is spanned by $x_1^{d-1},\ldots,x_r^{d-1}$. In the general case, $f$ is
    equivalent to $x_1^d+\cdots+x_r^d$ and two equivalent polynomials
    have the same number of essential variables.
  \end{example}
  The next proposition is a 
  consequence
  of the above variable minimization algorithm.
  The input $f$ to the algorithm of Proposition~\ref{prop:fewdecomp}
    can be described by an arithmetic circuit like in~\cite{kayal11}
  or more generally by a black box. Here we assume (in contrast with Sections~\ref{sec:random} and~\ref{sec:deter}) that we have access to an oracle for
    the factorization of univariate polynomials.
    This is a prerequisite for running Kaltofen's factorization algorithms
    for the arithmetic circuit~\cite{Kaltofen89} and black box models~\cite{KalTra90}.
  \begin{proposition} \label{prop:fewdecomp}
    There is a randomized polynomial time algorithm that decides
    whether a homogeneous polynomial $f(x_1,\ldots,x_n)$ can be written
    as in~(\ref{eq:linind})
  as a sum of powers of linearly independent linear forms.
  \end{proposition}
  \begin{proof}
    First compute the number $r$ of essential variables in $f$ using the
      randomized algorithm from~\cite{kayal11}, and make the corresponding change of variables
    to obtain a polynomial $g(x_1,\ldots,x_r)$. Then test whether
    $g$ is equivalent to $x_1^d+\cdots+x_r^d$ using the equivalence
    algorithm from~\cite{kayal11}. 
    { 
    To prove that this algorithm is correct, we show that $f$ can be written as a sum of $r$ powers of linearly independent linear forms if and only if $g$ can be written in such a form. Let $A \in \kk^{n \times n}$ be an invertible matrix such that $f(Ax) = g(x_1, \ldots, x_r)$ for all $x \in \kk^n$. Suppose that $g(Bx) = x_1^d+\cdots+x_r^d$ for some invertible matrix $B \in \kk^{r \times r}$. If we denote $C = \begin{pmatrix}B & 0 \\ 0 & I \end{pmatrix} \in \kk^{n \times n}$ where $I \in \kk^{(n-r) \times (n-r)}$ is the identity matrix, then the matrix $AC$ is invertible and we have $f(ACx) = g\bigl(B(x_1,\ldots,x_r)^T\bigr) = x_1^d+\cdots+x_r^d$ for every $x \in \kk^n$. Conversely, suppose that there exists an invertible matrix $B \in \kk^{n \times n}$ such that $f(Bx) = x_1^d+\cdots+x_{r'}^d$. We have $r' = r$ by Example~\ref{ex:essvar}. Moreover, for all $x \in \kk^n$ we have $x_1^d+\cdots+x_r^d = f(Bx) = f(AA^{-1}Bx) = g\bigl((A^{-1}Bx)_1,\ldots,(A^{-1}Bx)_r\bigr)$. Thus, by setting $x_{r+1} = \ldots = x_{n} = 0$, we get $x_1^d+\cdots+x_r^d = g(Cx)$ for all $x \in \kk^r$, where $C \in \kk^{r \times r}$ is the submatrix of $A^{-1}B$ obtained by taking the first $r$ rows and columns. To show that $C$ is invertible, suppose that this is not the case. Then, there exists an invertible matrix $D \in \kk^{r \times r}$ such that $D(0,\ldots,0,1)^{T} \in \ker C$. Hence, for every $x \in \kk^r$ we have $g(CD(x_1,\dots,x_{r-1},0)^T) = g(CDx) = h(Dx)$, where $h(x) = x_1^d+\cdots+x_r^d$. In particular, $h$ has less than $r$ essential variables, which gives a contradiction with Example~\ref{ex:essvar}. Therefore, the matrix $C$ is invertible.
    }
   \end{proof}
  In this algorithm it is essential to compute the number of essential variables
  in~$f$ before calling the equivalence algorithm from~\cite{kayal11}.
  Indeed, this algorithm is based on the factorization of the Hessian determinant
  of~$f$; but   $H_f$ is identically 0 for any polynomial
  with fewer  than $n$ essential variables. Hence 
  looking at $\det H_f$
  does not yield any useful information for $r <n$.

  \begin{remark}
    We can minimize the number of variables of a degree 3 form 
    $f(x_1,\ldots,x_n)$ in deterministic polynomial time 
    using dense linear algebra.
    Indeed, as pointed out in Section~\ref{sec:depend} this is true more generally for the POLYDEP problem with inputs that are verbosely given as sums of monomials.\footnote{Variable minimization for forms of degree 3 is also studied in Saxena's thesis~\cite[Proposition~3.1]{Saxenathesis}. He attributes the corresponding deterministic algorithm to Harrison~\cite{harrison75}.}
  Combining this observation with the deterministic equivalence algorithm
    from Section~\ref{sec:deter} we obtain{ , as in Proposition~\ref{prop:fewdecomp},}
    a deterministic algorithm to decide whether a degree 3 form can be written
    as in (\ref{eq:linind}) as a sum of cubes of linearly independent
    (real or complex) linear forms. 
  \end{remark}

    {      The second result of this section is the following derandomization of Kayal's algorithm for finding the number of essential variables, under the assumption that the input polynomial is a sum of powers of independent linear forms:}
  \begin{theorem}\label{thm:essderand}
    Let $f(x_1,\ldots,x_n)$ be a homogoneous polynomial of degree $d$
    and let $\{a_1,\ldots,a_k\}$ be the hitting set of Theorem~\ref{th:hitfew}
    corresponding to polynomials of degree $d-1$ in $n$ variables
    (recall that it is of size {  $O(n^2d)$}).
    Let $f_j$ be the partial derivative $\partial f / \partial x_j$.
    We consider like in Section~\ref{sec:depend} the matrix
    $M=(f_j(a_i))_{1 \leq i \leq k, 1\leq j \leq n}$.

    If $f$ can be written as in (\ref{eq:linind}) as a sum of $r$ powers
    of linearly independent linear forms then $\ker M = (\partial f)^{\perp}$.
    In particular, the number of essential variables of such an $f$
    can be computed deterministically  from a black box for $f$
    by the formula:
    $r=n - \dim \ker M$.
  \end{theorem}
  \begin{proof}
    We recall that a black box for $f_j$ can be easily obtained from a black
    box for $f$ by polynomial interpolation. It therefore remains to show
    that $\ker M = (\partial f)^{\perp}$.
      By Proposition~\ref{prop:hitequiv} it suffices to show
    that $\{a_1,\ldots,a_k\}$ is a hitting set
    for $\mathrm{Span}({\partial f})$.
    This is clear from the definition of $\{a_1,\ldots,a_k\}$ since
    the elements of $\mathrm{Span}({\partial f})$ can be written as
    linear combinations of at most $r$ 
    $(d-1)$-th powers
    of linearly independent linear forms (namely, the same forms that appear
    in the decomposition of $f$).
   \end{proof}

  \subsection{Lie algebras and polynomial factorization}
  \label{sec:Lie}

  One can associate to a polynomial $f \in K[X_1,\ldots,X_n]$
  the group of invertible
  $n \times n$ matrices $A$ that leave $f$ invariant, i.e., such that
  $f(Ax)=f(x)$. One can in turn associate to this matrix group its Lie algebra.
  This is a linear subspace of $M_n(K)$,
  which we call simply ``the Lie algebra of $f$.''
  It turns out that elements of this Lie algebra correspond to linear
  dependencies between the $n^2$ polynomials
  $x_j \frac{\partial f}{\partial x_i}$.
  A proof can be found in~\cite[Section~7.2]{kayal2012affine}, and we will take 
  this characterization
  as our definition of the Lie algebra for the purpose of this paper:
  \begin{definition} \label{def:lie}
    The Lie algebra of a polynomial $f \in K[x_1,\ldots,x_n]$ is the subspace
    of all matrices $C \in M_n(K)$ that satisfy the identity:
  $$\sum_{i,j \in [n]} c_{ij} x_{j} \frac{\partial f}{\partial x_{i}} = 0 \, .$$
  \end{definition}
  A randomized algorithm for the computation of the Lie algebra
  was given in \cite{kayal2012affine}, with applications to the reconstruction
  of affine projections of polynomials.
  In this section we study the deterministic computation of Lie algebras
  of polynomials,
  a topic which has not been studied in the literature as far as we know.
  
  The Lie algebra of a homogenous polynomial $f$ consists of all matrices
  of $M_n(K)$ if and only if $f$ is identically 0. This shows that one
  cannot hope to compute the Lie algebra in deterministic polynomial
  time without derandomizing Polynomial Identity Testing
  (and these two problems are in fact equivalent
  in the black box setting by Proposition~\ref{prop:hitequiv}).
  Nevertheless, it makes sense to search for deterministic
  algorithms for specific classes of polynomials.
  We take a first step in this direction in Theorem~\ref{th:deterlie},
  for polynomials that factor
  as products of linear forms. Taking again our cue from
  Proposition~\ref{prop:hitequiv}, we will do this by constructing
  a hitting set for a related family of polynomials.
  As it turns out, it is convenient to first design a hitting set
  for a certain family of ``simple'' rational functions.
  Those are defined as follows:
  \begin{definition} \label{def:simplerat}
    Let $\pnum_{1}, \dots, \pnum_{m} \in \cc^{n}, \pdenom_{1}, \dots, \pdenom_{m} \in \cc^{n} \setminus \{0\} $ be a collection of $2m$ vectors in $\cc^{n}$.
Furthermore,
    let $\hypar 
    = \bigcup_{i = 1}^{m} \{x \in \cc^{n} \colon \langle \pdenom_{i}, x \rangle = 0\}$.

    We associate to this collection of $2m$ vectors an oracle which for any
    $x \in \cc^{n}$ returns the value
\begin{equation} \label{eq:simplerat}
f(x) = 
\begin{cases}
\displaystyle{\sum_{i = 1}^{m} \frac{\langle \pnum_{i}, x\rangle}{\langle \pdenom_{i}, x \rangle}} & \text{if $x \notin \hypar$,} \\[0.6cm]
\NaN &\text{otherwise.}
\end{cases}
\end{equation}
  \end{definition}
  In  the commutative setting,
  there does not seem to be a lot of literature on rational identity testing
   (there is however the deep result that
  rational identity testing can be done in deterministic polynomial time
  in the non-commutative setting~\cite{garg16}).
  For (commutative) arithmetic circuits with divisions, deterministic
  rational identity testing
  is easily seen to be equivalent to PIT for ordinary
  (division free) arithmetic circuits. Nevertheless, it makes sense to
  investigate it for specific families of rational functions such
  as those in Definition~\ref{def:simplerat}.
  \begin{remark} \label{rem:domain}
    According to the above definition, the oracle returns NaN
    when we evaluate $f$ on a point $x$ where $q_i(x)=0$,
    and this remains true even if the corresponding vector $p_i$
    is equal to 0. This convention is useful for the proof
    of Proposition~\ref{hset} below.
   \end{remark}

  For every $n, k \in \nn$, let $\pconf(n,k) \subset \cc^{n}$ be the set of $k(n-1)+1$ points defined as
\begin{align*}
  \pconf(n,k) = 
  \{&(1,1,\dots,1), (1,2,2^{2},\dots,2^{n-1}), \\ 
&\dots, (1,k(n-1)+1,(k(n-1)+1)^{2},\dots,(k(n-1)+1)^{n-1}) \} \, .
\end{align*}
Recall from Lemma~\ref{lem:moment} that these points  cannot be all contained
in a union of $k$ hyperplanes since they lie on the moment curve.
Moreover, for every $m,n \ge 1$ let 
\[
\hset(m,n) = 
\{u + \lambda v \colon v \in \pconf(n,m), u \in \pconf(n,m^{2}) \cup \{0\}, \lambda \in [2m+1] \} \, .
\]
The next result shows that the set $\hset(m,n)$ is a hitting set for the rational functions of Definition~\ref{def:simplerat}.

\begin{proposition}\label{hset}
  The function $f(x)$ in~(\ref{eq:simplerat}) is equal to $0$
  for every $x \notin \hypar$ if and only if $f(x) \in \{0, \NaN\}$
  for every $x \in \hset(m,n)$.
\end{proposition}
\begin{proof}
  The ``only if'' implication is trivial. To prove the other direction, suppose that $f(x) \in \{0, \NaN\}$ for every $x \in \hset(m,n)$. First we observe that it is enough to assume that $\{\pdenom_{1}, \dots, \pdenom_{m}\}$ are pairwise linearly independent. Indeed, if $\pdenom_{i} = \mu \pdenom_{j}$ for some $i,j \in [m]$ and $\mu \in \cc \setminus \{0\}$, then we can put $\tilde{\pnum}_{i} =
  \pnum_{i} + \mu \pnum_{j}$, replace $\pnum_{i}$ by $\tilde{\pnum}_{i}$ and forget $\pnum_{j}$ and $\pdenom_{j}$. This does not change the function $f$ (in particular, by Remark~\ref{rem:domain} the domain of definition of $f$ is unchanged).
  By repeating this procedure, we can write $f$ in such a way that the denominators are pairwise linearly independent (and their number $m$ does not increase).

  From now on, we assume that $\{\pdenom_{1}, \dots, \pdenom_{m}\}$ are pairwise linearly independent. By Lemma~\ref{lem:moment}, there exists $v \in \pconf(n,m)$ such that $\langle \pdenom_{i}, v \rangle \neq 0$ for all $i \in [m]$. For every $i \in [m]$, denote $a_{i} = 
  \langle \pnum_{i}, v \rangle / \langle \pdenom_{i}, v \rangle \in \cc$. We will show that $\pnum_{i} = a_{i} \pdenom_{i}$ for all $i \in [m]$. To do so, suppose that there exists at least one $i$ such that $p_{i} - a_{i} q_{i} \neq 0$. For every pair $(i,j) \in [m]^{2}$ such that $i \neq j$ let $d_{ij} = 
  \langle \pdenom_{i}, v \rangle / \langle \pdenom_{j}, v \rangle$.
  Using Lemma~\ref{lem:moment} 
  one more time, there exists $u \in \pconf(n,m^{2})$
  satisfying the following two conditions:
\begin{itemize}
\item[(i)]  $\langle p_{i} - a_{i} q_{i}, u \rangle \neq 0$ for every $i \in [m]$ such that $p_{i} - a_{i} q_{i} \neq 0$;
\item[(ii)] $\langle \pdenom_{i} - d_{ij} \pdenom_{j}, u \rangle \neq 0$ for every $(i,j)$ such that $i \neq j$. (We note that $\pdenom_{i} - d_{ij} \pdenom_{j} \neq 0$ because the $\{\pdenom_{i}\}_{i}$ are pairwise linearly independent.)
\end{itemize}
  Let $b_{i} = 
  \langle p_{i} - a_{i} q_{i}, u \rangle$ for all $i \in [m]$ and consider the
  univariate function $g(\lambda)$ defined as $g(\lambda) = 
  f(u + \lambda v)$. Note that for every $\lambda \in \cc$ such that $u + \lambda v \notin \hypar$ we have
\begin{align*}
g(\lambda) = \sum_{i = 1}^{m}\frac{\langle p_{i}, u + \lambda v \rangle}{\langle q_{i}, u + \lambda v \rangle} = \sum_{i = 1}^{m} a_{i} + \sum_{i = 1}^{m} \frac{b_{i}}{\langle q_{i}, u \rangle + \lambda \langle q_{i}, v \rangle} \, .
\end{align*}
Observe that the function $g(\lambda)$ attains the value $\NaN$ for at most $m$ values of~$\lambda$. Furthermore, since we assumed that $\langle \pdenom_{i} - d_{ij} \pdenom_{j}, u \rangle \neq 0$, the functions $\lambda \mapsto \langle q_{i}, u \rangle + \lambda \langle q_{i}, v \rangle$ have distinct zeros. In particular, if $b_{i} \neq 0$, then 
$|g(\lambda)|$ approaches $+\infty$ 
as $\lambda$ approaches $-\langle \pdenom_{i}, u \rangle / \langle \pdenom_{i}, v \rangle$. 
Since we assumed that at least one $b_{i}$ is nonzero, it follows that the function $g(\lambda)$ attains some values not in $\{0 ,\NaN\}$. Moreover, $g(\lambda)$ has at most $m$ zeroes, because it can be written in the form $g(\lambda) = P(\lambda)/Q(\lambda)$ where $P, Q$ are nonzero polynomials of degree at most $m$. In particular, at least one of the values $g(1), \dots, g(2m + 1)$ does not belong to $\{0 ,\NaN\}$, contradicting our assumption. Therefore, we have $\pnum_{i} = a_{i} \pdenom_{i}$ for all $i \in [m]$. In particular, for every $x \notin \hypar$ we have $f(x) = a_{1} + \dots + a_{m}$.
To conclude, we observe that $f(v) = 0$ since $v \in \hset(m,n)$.
It follows that $f(x) = a_{1} + \dots + a_{m} = 0$ for all $x  \notin \hypar$.
\end{proof}

\begin{remark} \label{rem:hset}
  One can derive from the above proof a syntactic characterization
  of the rational functions in Definition~\ref{def:simplerat}
  that are identically 0. Namely, assuming that the $q_i$ are pairwise
  linearly independent, the following condition is necessary and sufficient:
  there exist constants $a_1,\ldots,a_m \in \cc$
  such that $a_1+\ldots+a_m = 0$ and $p_i = a_i q_i$ for all $i=1,\ldots,m$.
\end{remark}
Let $P(x) \in \cc[X_{1}, \dots, X_{n}]$ be a polynomial that factors as a product of linear forms, i.e., $P(x) = \langle \pdenom_{1}, x \rangle  \langle \pdenom_{2}, x \rangle \dots  \langle \pdenom_{d}, x \rangle$ for some vectors $\pdenom_{1}, \dots, \pdenom_{d}$ in $\cc^{n}$. From Proposition~\ref{hset} we can derive the
following characterization of the Lie algebra of $P$.
\begin{theorem} \label{th:deterlie}
Let $P(x) \in \cc[X_{1}, \dots, X_{n}]$ be a polynomial of degree $d \geq 1$ that factors as a product of linear forms. Then, a matrix $C \in \cc^{n \times n}$ belongs to the Lie algebra of $P$ if and only if
\begin{equation} \label{eq:lie}
\sum_{i,j \in [n]} c_{ij} x_{j} \frac{\partial P}{\partial x_{i}}(x) = 0
\end{equation}
for every $x \in \hset(d,n)$. In particular, a basis of the Lie algebra
can be computed deterministically in polynomial time
with black box access to $P$.
\end{theorem}
\begin{proof}
  The ``only if'' direction follows immediately from Definition~\ref{def:lie}.
  To prove the opposite implication, {  let us now assume that (\ref{eq:lie}) holds for every  $x \in \hset(d,n)$. 
  We need to show that $C$ belongs to the Lie algebra of $f$. }
   Let us write $P(x) = \langle \pdenom_{1}, x \rangle  \langle \pdenom_{2}, x \rangle \dots  \langle \pdenom_{d}, x \rangle$
  where the $q_i$ are nonzero vectors, and consider the function
\[
f_{C}(x) = 
\begin{cases}
\displaystyle{\frac{1}{P(x)}\sum_{i,j \in [n]} c_{ij} x_{j} \frac{\partial P}{\partial x_{i}}(x)} & \text{if $x \notin \hypar$,} \\[0.6cm]
\NaN &\text{otherwise.}
\end{cases}
\]
Here 
$\hypar$ denotes the union of the $m$ hyperplanes $\{x \in \cc^{n} \colon \langle \pdenom_{i}, x \rangle = 0\}$ as in Definition~\ref{def:simplerat}.
Note that we have $f_{C}(x) \in \{0, \NaN\}$ for every $x \in \hset(d,n)$. Furthermore, observe that for every $x \notin \hypar$ we have
\begin{align*}
\frac{1}{P(x)}\sum_{i,j \in [n]} c_{ij} x_{j} \frac{\partial P}{\partial x_{i}}(x) &= \sum_{i,j \in [n]}c_{ij}x_{j}\sum_{k \in [d]}\frac{\pdenom_{ki}}{\langle \pdenom_{k}, x \rangle} = \sum_{k \in [d]}\frac{\sum_{i,j \in [n]} c_{ij}\pdenom_{ki}x_{j}}{\langle \pdenom_{k}, x \rangle} \, .
\end{align*}
Since this rational fraction is of form~(\ref{eq:simplerat}), we can
apply Proposition~\ref{hset} and conclude that 
$f_{C}(x)$ is equal to zero for every $x \notin \hypar$.
This implies that $\sum_{i,j \in [n]} c_{ij} x_{j} \frac{\partial P}{\partial x_{i}}(x) = 0$ for all $x \notin \hypar$.
By continuity we obtain $\sum_{i,j \in [n]} c_{ij} x_{j} \frac{\partial P}{\partial x_{i}}(x) = 0$ for all $x \in \cc^{n}$, which implies that $C$ belongs to the Lie algebra of $P$.

Let us now turn to the second part of the theorem. It is well known that
a black box for $\partial P / \partial x_i$ can be constructed from a black box
for $P$ (by interpolating $P$ on a line). By the first part, the determination
of the Lie algebra therefore boils down to the resolution of a system
of $|\hset(d,n)|$ linear equations in $n^2$ variables.
\end{proof}

\begin{remark}
  We have stated Proposition~\ref{hset} and Theorem~\ref{th:deterlie}
  for the field of complex numbers only because the proof of Proposition
  \ref{hset} uses the absolute value.  Nevertheless, it follows from general principles that  these two results apply 
  to any field $K$ of characteristic~0. Indeed, $K$ can be embedded in
  an algebraically closed field $\overline{K}$ which must satisfy the
  same first order formulas as $\cc$.
\end{remark}

{ 
For our final derandomization results we will need to perform simultaneous diagonalization in polynomial time over $\qq$.
\begin{proposition} \label{prop:detersimdiag}
There is a polynomial time deterministic algorithm which takes as input a tuple $(A_1,\ldots,A_k)$ of matrices of size $n$
with rational entries, and:
\begin{itemize}
\item[(i)] decides whether $A_1,\ldots,A_k$ are simultaneously diagonalizable over $\qq$;
\item[(ii)] if they are, constructs an invertible matrix $T \in M_n(\qq)$ such that the $k$ matrices $T^{-1}A_iT$ are all diagonal.
\end{itemize}
\end{proposition}
This result is not particularly surprising but we could not find it in the literature.
\begin{proof}
For $k=1$ this is quite standard. We can for instance compute the characteristic polynomial of $A_1$, compute its roots (which should all be rational) and attempt to construct a basis of eigenvectors by solving the corresponding linear systems.

For $k>1$, we can check that the matrices are simultaneously diagonalizable by checking that each matrix is diagonalizable, and that the $A_i$ pairwise commute (Theorem~1.3.21 in~\cite{horn13}).
In this case, in order to construct a transition matrix $T$ which diagonalizes the $A_i$, 
we will use Lemma~\ref{lincomb2} to reduce to the case $k=1$. 
Namely, we will construct a finite set $S$ of points $(\alpha_2,\ldots,\alpha_n) \in \qq^{n-1}$, and for each $\alpha \in S$
we will:
\begin{enumerate}
\item Diagonalize $A_{\alpha}=A_1+\alpha_2 A_2 + \cdots + \alpha_k A_k$ as: 
$A_{\alpha} = T_{\alpha} D_{\alpha} T_{\alpha}^{-1}$, where $D_{\alpha}$ is diagonal and $T_{\alpha}$ invertible.
\item Check whether $T_{\alpha}^{-1} A_i T_{\alpha}$ is diagonal for all $i=1,\ldots,k$.
\end{enumerate}
When the $A_i$ are simultaneously diagonalizable, Lemma~\ref{lincomb2}.(ii) guarantees that this test will will succeed for at least one $\alpha \in S$ if $S$ is not included in a certain union of $n(n-1)/2$ hyperplanes. In order to avoid these hyperplanes we can proceed as in Section~\ref{sec:deter} and pick any set $S$ of $1+n(n-1)(n-2)/2$ points on the moment curve as per Lemma~\ref{lem:moment}.
\end{proof}
An alternative to the above algorithm can possibly be extracted from the proof of Theorem~1.3.21 in~\cite{horn13}.}

Let $f(x_1,\ldots,x_n)$ be a polynomial that can be written as
\begin{equation} \label{eq:factor}
  f(x)=\lambda l_1(x)^{\alpha_1}\cdots l_n(x)^{\alpha_n}
\end{equation}
where $\lambda$ is a constant, the $l_i$ are linearly independent linear forms {  and the exponents $\alpha_i$ are all nonzero.}
A randomized algorithm which finds such a factorization
from black box access to $f$ was proposed
in~\cite[Section~4]{koiran2018orbits}.
 {  This algorithm appealed to randomization for the computation of the Lie algebra of $f$ and also, following~\cite{kayal2012affine}, for simultaneous diagonalization. In this paper we have given deterministic algorithms 
 for these two tasks, in Theorem~\ref{th:deterlie} and Proposition~\ref{prop:detersimdiag} respectively.
 This leads to the following polynomial time deterministic factorization algorithm, 
  which we call
  the {\em derandomized Lie-algebraic factorization algorithm}
  (or {\tt DerandLie} for short):
  \begin{enumerate}
\item Compute a basis $B_1,\ldots,B_k$ of the  Lie algebra of $f$.
\item Reject if $k \neq n-1$, i.e., if the Lie algebra is not of dimension $n-1$.
  
\item Check that the matrices $B_1,\ldots,B_{n-1}$ commute and are 
  all diagonalizable over $\qq$.
  If this is not the case, reject.
  Otherwise, declare the existence of a factorization
  $f(x)= \lambda l_1(x)^{\alpha_1} \cdots l_n(x)^{\alpha_n}$
  where the linear forms $l_i$ are linearly independent
  and $\alpha_i \geq 1$ ($\lambda$, the $l_i$ and $\alpha_i$ will be determined
  in the last 3 steps of the algorithm).

  \item Perform a simultaneous diagonalization of the $B_i$'s, i.e.,
  find an invertible matrix $A$ such that the $n-1$ matrices
  $AB_iA^{-1}$ are diagonal.

\item At the previous step we have found a matrix~$A$ such that $g(x)=f(A^{-1}x)$ has a Lie algebra $\mathfrak{g}_g$ which is an $(n-1)$-dimensional
  subspace of the space of diagonal matrices.
  Then we compute the orthogonal of $\mathfrak{g}_g$, i.e., 
  we find a vector $\alpha=(\alpha_1,\dots,\alpha_n)$ such $\mathfrak{g}_g$
  is the space of matrices 
  $\diag(d_1,\ldots,d_n)$ satisfying $\sum_{i=1}^n \alpha_i d_i = 0$.
  We normalize $\alpha$ so that $\sum_{i=1}^n \alpha_i = d$.

\item   We must have $g(x)=\lambda.m$ where $\lambda \in {\qq}^*$ 
  and $m$ is the monomial $x_1^{\alpha_1} \cdots x_n^{\alpha_n}$ (in particular, $\alpha$ must be a vector with integral entries).
  We therefore have $f(x)=\lambda.m(Ax)$ and we output this factorization.
\end{enumerate}

  \begin{theorem} \label{th:derandlie}
  Let $f(x_1,\ldots,x_n)$ be a polynomial that can be written as in~(\ref{eq:factor}) as a product of powers of linearly independent linear forms, where $\lambda$ and 
  the coefficients of the linear forms are in $\qq$. From a black box for $f$, 
  the above {\tt DerandLie} algorithm computes this factorization deterministically in polynomial time.
  \end{theorem}
 See \cite[Section~4]{koiran2018orbits} for a correctness proof. As mentioned above, in order to 
 obtain a deterministic algorithm we appeal to Theorem~\ref{th:deterlie} in Step~1 and to Proposition~\ref{prop:detersimdiag} 
 in Step~4. In order to find the scaling factor $\lambda$ at step~6, we evaluate $f$ at a point $x$ where $f(x) \neq 0$.
 From Lemma~\ref{lem:moment}, we can find such a point deterministically by trying at most $1+n(n-1)$ on the 
 moment curve since we need to avoid the $n$ hyperplanes $l_i(x)=0$.
  }
  Note that  {\tt DerandLie} may fail if $f$ does not factor as a product of linear forms since this is a prerequisite of Theorem~\ref{th:deterlie}.
  The fact that {\tt DerandLie} fails on some inputs may seem at first sight
  like a weakness of the algorithm, but this is in fact unavoidable for
  {\em any} deterministic polynomial-time black box algorithm
  (see~\cite[Section~1.5]{koiran2018orbits} for details). 

  Recall from Section~\ref{subsec:kayal} that Kayal's algorithm for equivalence to a sum of powers relies
  on factorization into products of linear forms.
  If this factorization is performed with the {\tt DerandLie} algorithm,
  we obtain a deterministic version of Kayal's algorithm.
  Let us call {\tt LieEquivalence} this deterministic 
   equivalence algorithm.
  As our final result, we observe that {\tt LieEquivalence}
  will work correctly on {\em all} inputs due
  to the presence of the verification step in Kayal's algorithm:
  \begin{theorem}
    Let $f \in \qq[X_1,\ldots,X_n]$ be a homogeneous polynomial of degree~$d$
    given verbosely as a sum of monomials. The {\tt LieEquivalence} algorithm
    determines whether $f$ is equivalent over $\qq$
    to $P_d$,
    the ``sum of $d$-th powers'' polynomial from~(\ref{eq:sumofpowers}).
    If this is the case, it outputs an invertible matrix $A$
    with rational entries such that $f(x)=P_d(Ax)$. Moreover, for any fixed
    $d$ the algorithm runs in polynomial time in the Turing machine model.
   \end{theorem}
  \begin{proof}
    Since we are interested in equivalence over the field of rational numbers,
    we will run the 3-step algorithm from Section~\ref{subsec:kayal}
    with $K=\mathbb{K}=\qq$.
    First we establish the correctness of {\tt LieEquivalence}. If the
    algorithm accepts its input $f$, it explicitly finds at step 2 and step 3
    linearly independent linear forms $\ell_i \in \qq[x_1,\ldots,x_n]$
    such that $f=\sum_{i=1}^n \ell_i^d$. The algorithm's answer must therefore
    be correct in this case. Conversely, assume that such a decomposition
    exists. Then the Hessian determinant $H_f$ factors as
    $H_f = c \prod_{i=1}^n \ell_i^{d-2}$ where $c$ is a nonzero constant. 
    Since the $\ell_i$ are linearly independent, we are in the situation
    where $\tt DerandLie$ works correctly. Therefore, by Theorem~\ref{th:derandlie} we will  find the
    $\ell_i$ (or actually constant multiples $l_i$ of the $\ell_i$)
    at step 1 of the algorithm of Section~\ref{subsec:kayal}.
    Finally, the decomposition $f=\sum_{i=1}^n \ell_i^d$ is obtained at steps
    2 and 3.

    We now turn to the algorithm's complexity. Since $\tt DerandLie$ runs in
    polynomial time, the first step of the algorithm
    from Section~\ref{subsec:kayal} will also run in polynomial time.
    At step 2 we can afford to expand the powers $l_i^d$ as sums of monomials
    (this takes polynomial time for constant $d$), and then we find the constants $a_i$ by dense linear algebra. Finally, the extraction of $d$-th roots of rational numbers at step 3 also takes polynomial time.
   \end{proof}
{\small

{  \section*{Acknowledgements} We would like to thank Roger Horn for pointing out Problem~4.5.P4 in~\cite{horn13} and Peter Bürgisser for references to the tensor literature. The anonymous referee suggested several improvements in the presentation of the paper, and a significant simplification in the hitting set constructions of Section~\ref{sec:PIT}.}

\bibliographystyle{plain}

}

\end{document}